\documentclass[10pt]{amsart}
\usepackage{egastyle}

\usepackage{graphicx}
\usepackage{subcaption}
\usepackage{tikz,tikzcd}
\usepackage{verbatim}

\usepackage{hyperref}

\newcommand{\Tc}{{\mathcal{T}}}
\newcommand{\Fc}{{\mathcal{F}}}
\newcommand{\wt}{\widetilde}
\newcommand{\image}{{\mathrm{im}\,}}

\usepackage{amssymb,amsmath,amscd,amsthm,amsfonts,wasysym,mathrsfs,amsxtra}
\input xy
\xyoption{all}

\tikzcdset{scale cd/.style={every label/.append style={scale=#1},
    cells={nodes={scale=#1}}}}

\begin{document}

\title{Quantifying analogy of concepts via ologs and wiring diagrams}


\author[Jason Lo]{Jason Lo}
\address{Department of Mathematics \\
California State University, Northridge\\
18111 Nordhoff Street\\
Northridge CA 91330 \\
USA}
\email{jason.lo@csun.edu}

\begin{abstract}
We build on the theory of ontology logs (ologs) created by Spivak and Kent, and define a notion of wiring diagrams.  In this article, a wiring diagram is a finite directed labelled graph.  The labels correspond to types in an olog; they can also be interpreted as readings of sensors in an autonomous system.  As such, wiring diagrams can be used as a framework for an autonomous system to form abstract concepts.  We show that the graphs underlying skeleton wiring diagrams form a category.  This allows skeleton wiring diagrams to be compared and manipulated using techniques from both graph theory and category theory.  We also extend the usual definition of graph edit distance to the case of wiring diagrams by using operations only available to wiring diagrams, leading to a metric on the set of all skeleton wiring diagrams.  In the end, we give an extended example on calculating the distance between two concepts represented by wiring diagrams, and explain how to apply our framework to any application domain.
\end{abstract}

\maketitle

\tableofcontents



\section{Introduction}

Analogical reasoning is a technique that humans often use in problem-solving.  We apply analogical reasoning when we are dealing with a problem in a new situation, where the problem bears some resemblance to ones that we have solved before.  In a more mundane, everyday situation, this could mean figuring out how to catch a train in a city we have never been in - from past experiences, we might gather that a plan that has a good chance of working is to buy a ticket first, and then find the right platform for the train.  Of course, the details as to how to purchase a ticket and how to find a specific train platform would vary from location to location, but as long as we attempt to execute steps that are `close' to the plan we have in mind, or `similar' to steps that we had taken in the past for getting onto the right train, our plan has a reasonable chance of succeeding.  We also use analogical reasoning in situations that are much more nuanced and complex, such as in scientific research, in setting government policies, or in legal arguments in the courtroom.  The recognition of analogues is also important in human experiences, such as in art, literature, and music.

In order to apply analogical reasoning in designing a problem-solving framework by autonomous systems, we first need a way to recognize when two concepts are similar.  In other words, we need a way to \emph{quantify} the similarity between two concepts.  Although there are existing methods for quantifying analogy such as word-embedding algorithms \cite{petersen-van-der-plas-2023-language,drozd-etal-2016-word}, these methods use statistical or probabilistic techniques that rely on having enough data, which is not always how humans recognize analogies.  Humans can recognize  analogies by determining the internal structures of concepts and by categorizing concepts \cite{KRAWCZYK2018227}.

In this article, we describe a mathematical approach to quantifying the analogy of two concepts.  Our approach builds on the idea of ontology logs, or \emph{ologs}, first coined in a paper by Spivak and Kent in 2012 \cite{SpivakKent}.  Ologs gives a direct connection between data (such as those received from sensors in an autonomous machine) and category theory (a part of mathematics that studies the relations among objects, independent of context).  In particular, we define the notion of  \emph{wiring diagrams}, which are labelled directed graphs satisfying certain axioms.  Wiring diagrams can be used to represent processes that occur over time, and hence can be used to represent complex concepts.  The labels of wiring diagrams correspond to concepts that appear as objects in ologs, whereas the underlying directed graphs of wiring diagrams themselves form a category.  Since ologs are themselves categories, we can compare wiring diagrams using both the underlying graph structures and the underlying categorical structures, thus quantifying the similarity between two concepts (as long as they are represented as wiring diagrams). 

\paragraph[Why ologs?]  Ologs give a way to organize concepts and the relations among them.  Every olog is a category in the sense of category theory, which is a language that is used across major branches of modern mathematics.  In addition, every olog is  associated to a database schema.  As a result, ologs provide a bridge via which tools from different areas of mathematics such as algebra, topology and geometry can be used to organize and understand data.  Indeed, ologs have been applied to fields such as biology \cite{spivak2011category,WuYY},  linguistics \cite{PerezSpivak}, materials design and manufacturing \cite{giesa2012category,brommer2015categorical}, among others.

\paragraph[Wiring diagrams]\label{para:intro-WD} Wiring diagrams have long been used to represent the various components and connections in an electric circuit.  Mathematically, wiring diagrams can  be defined and studied as operads \cite{SR-WDdisc,spivak2013operad,VSL,yau2018operads}.  In this article, we settle for a more simplistic  definition of a wiring diagram -  roughly speaking, a directed graph with labels that correspond to objects in an olog.  Of course, it would be interesting to work out the precise connections between the operadic approach towards studying discrete-time processes taken in \cite{SR-WDdisc} and our approach.

\paragraph[Outline of the article]  In Section \ref{sec:ologsreview}, we give a brief example to illustrate the basic terminology from the theory of ologs.  In Section \ref{sec:distbtwconcepts-start}, we review the idea that any classification scheme for concepts that involves a series of `multiple-choice' questions - such as the dichotomous identification key for insects that students may learn in high school - can be constructed as an olog.  We also list in Section \ref{para:tipspopulateolog} basic ideas for populating an olog for use in a specific application domain.  In Section \ref{sec:quantifysim-1}, we demonstrate how intangible concepts such as relations among physical entities can  be represented by ologs.  We also point out the obvious fact, that since an olog has an underlying graph, one can use any reasonable metric on graphs such as the shortest-distance metric to define a distance between any two objects in an olog.  

In Section \ref{sec:WDforprocesses}, we give the mathematical definition of our version of wiring diagrams.  We prove a couple of elementary mathematical properties of wiring diagrams in Lemmas \ref{lem:WDgraph-1} and \ref{lem:WDgraphsRaxal}.  Then, through an example in \ref{eg:buyingcoffee}, we illustrate how the occurence of a concept that is represented by a wiring diagram can be detected using data collected from sensors over time.  

In Section \ref{sec:catofWDs}, we define the notion of skeleton wiring diagram graphs, or skeleton WD graphs.  These are directed graphs that are the underlying graphs of skeleton wiring diagrams.  We show in Section \ref{para:catskWDgs} that skeleton WD graphs with a common set $V$ of vertices form a category $\mathcal{R}(V)$.  As such, morphisms in this category give us new ways to compare wiring diagrams, i.e.\ new ways to compare concepts, that were not possible if one only considered traditional edit operations on graphs.  Then in Section \ref{sec:distanceedit-WD}, we use morphisms in the category $\mathcal{R}(V)$ to extend the usual definition of elementary edit operations on graphs to the case of wiring diagrams.

In Section \ref{sec:compareanalogy}, we give an extended example showing how the ideas in Sections \ref{sec:WDforprocesses} and \ref{sec:catofWDs} can be implemented to quantify the analogy between two concepts.  We chose the concepts of an `electric car charging station' and a `bus'.  Both are physical entities that are capable of changing a characteristic of another physical entity when a particular relation is satisfied.  We explain how to define relevant sensors and  ologs, and then how to construct wiring diagrams that can be used to represent the two concepts `electric car charging station' and `bus'.  Then, we show how elementary edit operations on wiring diagrams can be used to define a `distance' between these two concepts, thus achieving our goal of quantifying analogy between concepts that are represented by wiring diagrams.  In \ref{para:summaryofcalc}, we give a list of steps that one can follow in order to implement the ideas in this article to quantify analogy in any application domain of the reader's choice.  Finally, we end the article with a brief discussion on future directions in Section \ref{sec:futureDs}.

\paragraph[Acknowledgements] The author would like to thank David I Spivak and Nima Jafari for helpful comments on earlier versions of the manuscript.    This article is based upon work supported by a DARPA Young Faculty Award (number D21AP10109-02) and an Air Force Office of Scientific Research grant (number FA9550-24-1-0268).

\section{Ologs and data}\label{sec:ologsreview}

We assume the reader has a rudimentary knowledge of the language of category theory and ologs, including the concept of fiber product (or `pullback') in a category.  Basic concepts in  category theory can be found in books such as \cite{SpivakCTS,walters_1992,Awodey} which are aimed at a general audience, or  \cite{maclane:71} which is aimed at  mathematicians.  On the other hand, a quick introduction to ologs can be found in the paper by Spivak and Kent \cite[Sections 1-3]{SpivakKent} in which ologs were first defined.  In this section, we briefly recall some basic terminology from the theory of ologs by way of an example.  We will also assume throughout this article that all the categories that arise are small.

\paragraph[Example]  An olog is a category in which
\begin{itemize}
\item[(i)] Each object and each arrow is labeled with text to indicate their meaning.
\item[(ii)] Each arrow represents a relation that corresponds to a function.
\end{itemize}
Property (ii) is one of the main differences between the olog approach and the knowledge graph approach to knowledge representation.  It implies that any two consecutive arrows can be composed (because any two functions where the codomain of one equals the domain of the other can be composed), which is one of the requirements of a category.  Since an olog is a category, which has an underlying graph, both graph-theoretic and category-theoretic tools are available when dealing with ologs.

For example, Figure \ref{fig:olog1}  is an olog with three objects and two arrows.  In this olog, given any pair $(p,c)$ where $p$ is a person, $c$ is a car, and $p$ owns $c$, the arrow $p$ (resp.\ $c$) represents the operation that ``forgets'' the information $c$ (resp.\ $p$) and only remembers $p$ (resp.\ $c$); using a more mathematical language, we can think of $p$ and $c$ as the first and second projections from the ordered pair $(p,c)$, respectively.
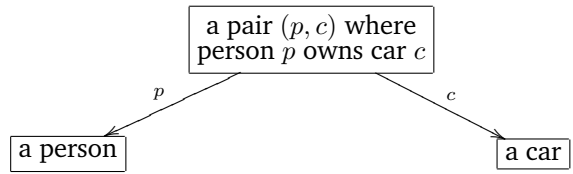
\begin{figure}[h]
  \centering
   \[
    \xymatrix{
   & *+[F]{\txt{a pair $(p,c)$ where \\
   person $p$ owns car $c$}} \ar[dl]_(.6)p \ar[dr]^(.6)c & \\
     *+[F]{\txt{a person}}  & &  *+[F]{\txt{a car
  }}
    }
    \]
    \caption{An example of an olog.}
    \label{fig:olog1}
\end{figure}
We will follow the language in \cite{SpivakKent}, and refer to objects in an olog as \emph{types},  arrows in an olog as \emph{aspects}, and write  $\lceil$something$\rceil$ instead of $\xymatrix{ *+[F]{\txt{something}}}$ when we want to refer to a type in the main text of this article.  We will also use the notation $\dashv$p$\vdash$ when we refer to an aspect labelled as $p$.

We can regard an olog as a representation of the internal knowledge of an autonomous system.  In this article,  we will focus on ologs where, for every point $t$ in time, each type $\Tc$ in an olog corresponds to a table $\Fc_t(\Tc)$ containing all  instances of the concept $\Tc$ known to the autonomous system.  For example, suppose at the point $t$ in time, there are 5 persons known to the system, say  Adam, Betty, Carlos, Dolly, Eric, and  4 cars known to the system, say (labeled by their license plate numbers) 9XL3A, 2HA1T, 8WT9R, 5RV3Q, 6PK7M.  Suppose also that at time $t$, the system is aware that Betty owns the car 9XL3A, Carlos owns 6PK7M, and that Dolly owns 5RV3Q.  Then the three types in Figure \ref{fig:olog1} would yield the three tables in Table \ref{tab:tables1}, which together constitute a simple database containing data that corresponds to the concepts in the olog in Figure \ref{fig:olog1}.  The rows of these tables are called the \emph{instances} of the corresponding type.  The arrows $p$ and $c$ in Figure \ref{fig:olog1} then yield the operations $\Fc_t(p), \Fc_t(c)$ which, respectively, sends each pair in the middle table to  either the first coordinate (`person') or the second coordinate (`car').  For example, $\Fc_t(p)$ would send the pair `(Betty, 9XL3A)' to `Betty' while $\Fc_t(c)$ would send it to `9XL3A'.  In the language of category theory, we say that $\Fc_t$ is a functor.  (See \cite{SpivakKent,SpivakCTS} for more on the connections between ologs and database schemas.)

\bigskip

\begin{table}
\begin{subtable}[c]{0.3\textwidth}
\centering
\begin{tabular}{|c|}
\hline
 Person  \\
\hline Adam \\
 Betty  \\
 Carlos \\
 Dolly \\
 Eric \\
\hline
\end{tabular}
\subcaption{Table for the type $\lceil$a person$\rceil$ }
\end{subtable}
\begin{subtable}[c]{0.3\textwidth}
\centering
\begin{tabular}{|c|}
\hline
 Pair $(p,c)$ where \\ person $p$ owns car $c$  \\
\hline (Betty, 9XL3A) \\
  (Carlos, 6PK7M) \\
 (Dolly, 5RV3Q) \\
\hline
\end{tabular}
\subcaption{Table for the type $\lceil$a pair $(p,c)$ where person $p$ owns car $c$$\rceil$}
\end{subtable}
\begin{subtable}[c]{0.3\textwidth}
\centering
\begin{tabular}{|c|}
\hline
 Car  \\
\hline 9XL3A \\
 2HA1T  \\
 8WT9R \\
 5RV3Q \\
 6PK7M \\
\hline
\end{tabular}
\subcaption{Table for the type $\lceil$a car$\rceil$}
\end{subtable}
\caption{\label{tab:tables1} Tables that together form a database.}
\end{table}

\bigskip

\section{Using ologs to classify concepts}\label{sec:distbtwconcepts-start}

Since an olog is defined to be a category in the mathematical sense, all the tools in category theory apply to ologs.  The fiber product construction in category theory, for example, offers a way to make precise the `overlap' or intersection of two concepts.  Using fiber products, any scheme that classifies a collection of concepts via a series of multiple-choice questions, such as a flow chart for identifying insects that are often taught in high school biology  can be incorporated into an olog.

\paragraph[From a classification scheme to an olog] \label{para:classifytoolog} Consider the following series of questions, which might be used as part of a scheme that tries to classify different types of transport vehicles:

\begin{itemize}
\item Question 1: Does it have wheels? Answer:
  \begin{itemize}
  \item[1.] Yes
  \item[2.] No
  \end{itemize}
  \item Question 2: What power source does it use? Answer:
  \begin{itemize}
  \item[1.] Human power.
  \item[2.] Electricity.
  \item[3.] Gas.
  \end{itemize}
\end{itemize}

Recording the answers to these questions when we apply them to a bicycle and a gas-powered car, we obtain the table in Table \ref{tab:MQeg}.

\begin{table}[h!]
\centering
\begin{tabular}{c | c c}
  Question & Bicycle & Gas-powered car \\ \hline
  Question 1 & 1 & 1 \\
  Question 2 & 1 & 3 \\
\end{tabular}
\caption{}
\label{tab:MQeg}
\end{table}

The answers recorded in Table \ref{tab:MQeg} allow us to associate the concepts `bicycle' and `gas-powered car' to different vectors in $\mathbb{R}^2$
\begin{align*}
  \text{bicycle} &\mapsto (1,1) \\
  \text{gas-powered car} &\mapsto (1,3).
\end{align*}
 The point here, however,  is that we can  construct an olog  that distinguishes $\lceil$a bicycle$\rceil$ and $\lceil$a gas-powered car$\rceil$ as two distinct types.  To do this, let us begin with the simple olog in  \eqref{fig:olog3}, which represents the fact `a bicycle is a human-powered vehicle'.

\begin{equation}\label{fig:olog3}
\begin{tikzcd}
	{\boxed{\text{a bicycle}}} && {\boxed{\begin{matrix}\text{a human-powered} \\ \text{vehicle}\end{matrix}}}
	\arrow["{\text{is}}", from=1-1, to=1-3]
\end{tikzcd}\end{equation}

The database schema corresponding to this olog consists of two tables: a table listing all known instances of a bicycle, and a table listing all known instances of a human-powered vehicle,  such that every instance that appears in the first table also appears in the second table.  Table \ref{tab:tables2} shows an example of such a schema.

\begin{table}[h]
\begin{subtable}[c]{0.3\textwidth}
\centering
\begin{tabular}{|c|}
\hline
 Bicycle  \\
\hline
 Bicycle 1  \\
 Bicycle 2 \\
 Bicycle 3 \\
\hline
\end{tabular}
\subcaption{Table for the type $\lceil$a bicycle$\rceil$ }
\end{subtable}
\begin{subtable}[c]{0.3\textwidth}
\centering
\begin{tabular}{|c|}
\hline
 Human-powered   \\
 vehicle \\
\hline
  Kayak 1 \\
  Bicycle 1\\
  Kayak 2 \\
  Bicycle 2\\
  Skateboard 1 \\
  Bicycle 3\\
\hline
\end{tabular}
\subcaption{Table for the type $\lceil$a human-powered vehicle$\rceil$}
\end{subtable}
\caption{An example of tables in a database schema for the olog in Figure \ref{fig:olog3}.}
\label{tab:tables2}
\end{table}

Intuitively, the concept of a human-powered vehicle represents the overlap of two different concepts: a transport vehicle, and the use of human power.  More concretely, a human-powered vehicle can be defined as a transport vehicle that uses human power as a source of power.  As a result, in the setting of ologs,  we can construct the type $\lceil$a human-powered vehicle$\rceil$ using a fiber product  as follows.  We begin with  the olog in  \eqref{fig:olog4}, where the vertical arrow is the function that takes any transport vehicle as its input, and gives its type of power source (e.g.\ human power, gas, electricity, etc.) as its output.

\begin{equation}\label{fig:olog4}
\begin{tikzcd}
	&& {\boxed{\begin{matrix}\text{a transport} \\ \text{vehicle}\end{matrix}}} \\
	\\
	{\boxed{\text{human power}}} && {\boxed{\begin{matrix}\text{a type of} \\ \text{power source}\end{matrix}}}
	\arrow["{\text{is}}", from=3-1, to=3-3]
	\arrow["{\text{has as power source}}", from=1-3, to=3-3]
\end{tikzcd}\end{equation}

The pullback of the olog in  \eqref{fig:olog4} is  \eqref{fig:olog5}, in which  the upper and the left arrows are newly generated in the pullback construction.  The upper  row  represents the fact `a human-powered vehicle is a transport vehicle', and the left  vertical arrow represents the same function as the right vertical arrow (in this case, every instance of a human-powered vehicle uses human power as its power source, and so the codomain of the left vertical arrow is simply `human power').

\begin{equation}\label{fig:olog5}
\begin{tikzcd}
	{\boxed{\begin{matrix}\text{a human-powered} \\ \text{vehicle}\end{matrix}}} && {\boxed{\begin{matrix}\text{a transport} \\ \text{vehicle}\end{matrix}}} \\
	\\
	{\boxed{\text{human power}}} && {\boxed{\begin{matrix}\text{a type of} \\ \text{power source}\end{matrix}}}
	\arrow["{\text{is}}", from=3-1, to=3-3]
	\arrow["{\text{has as power source}}", from=1-3, to=3-3]
	\arrow["{\text{is}}", from=1-1, to=1-3]
	\arrow[from=1-1, to=3-1]
\end{tikzcd}\end{equation}

Similarly, we can construct the concept $\lceil$a gas-powered vehicle$\rceil$ in an olog using a fiber product as in \eqref{fig:olog6}.

\begin{equation}\label{fig:olog6}
\begin{tikzcd}
	{\boxed{\begin{matrix}\text{a gas-powered} \\ \text{vehicle}\end{matrix}}} && {\boxed{\begin{matrix}\text{a transport} \\ \text{vehicle}\end{matrix}}} \\
	\\
	{\boxed{\text{gas}}} && {\boxed{\begin{matrix}\text{a type of} \\ \text{power source}\end{matrix}}}
	\arrow["{\text{is}}", from=3-1, to=3-3]
	\arrow["{\text{has as power source}}", from=1-3, to=3-3]
	\arrow["{\text{is}}", from=1-1, to=1-3]
	\arrow[from=1-1, to=3-1]
\end{tikzcd}\end{equation}

Since the right vertical arrows in  \eqref{fig:olog5} and \eqref{fig:olog6} coincide, both of these ologs can be incorporated into the larger olog in  \eqref{fig:olog2}.

\begin{equation}\label{fig:olog2}
\begin{tikzcd}
	{\boxed{\text{a bicycle}}} &&& {\boxed{\begin{matrix}\text{a human-powered} \\ \text{vehicle}\end{matrix}}} && {\boxed{\begin{matrix}\text{a transport} \\ \text{vehicle}\end{matrix}}} \\
	{\boxed{\begin{matrix}\text{a gas-powered} \\ \text{passenger car}\end{matrix}}} && {\boxed{\begin{matrix}\text{a gas-powered} \\ \text{vehicle}\end{matrix}}} \\
	&&& {\boxed{\text{human power}}} && {\boxed{\begin{matrix}\text{a type of} \\ \text{power source}\end{matrix}}} \\
	&& {\boxed{\text{gas}}}
	\arrow["{\text{is}}", from=3-4, to=3-6]
	\arrow["{\text{has as power source}}", from=1-6, to=3-6]
	\arrow["{\text{is}}", from=1-4, to=1-6]
	\arrow[from=1-4, to=3-4]
	\arrow["{\text{is}}"', from=2-3, to=1-6]
	\arrow["{\text{is}}"', from=4-3, to=3-6]
	\arrow["{\text{is}}", from=1-1, to=1-4]
	\arrow["{\text{is}}", from=2-1, to=2-3]
	\arrow[from=2-3, to=4-3]
\end{tikzcd}\end{equation}

The olog \eqref{fig:olog2}  `recovers' Question 2 in the multiple-choice classification scheme from the start of this subsection.  To find the answer when we apply  Question 2 to the concept `a bicycle', for example, we look among  types $B$ that correspond to the upper-left vertex of a fiber product diagram of the form in  \eqref{fig:olog8}, where $A$ is a type corresponding to a specific example of a power source, and find the particular $B$ such that there is an injection $\dashv$is$\vdash$ pointing from $\lceil$a bicycle$\rceil$ to $B$.

\begin{equation}\label{fig:olog8}
\begin{tikzcd}
	B && {\boxed{\begin{matrix}\text{a transport} \\ \text{vehicle}\end{matrix}}} \\
	\\
	A && {\boxed{\begin{matrix}\text{a type of} \\ \text{power source}\end{matrix}}}
	\arrow["{\text{is}}", from=3-1, to=3-3]
	\arrow[from=1-1, to=1-3]
	\arrow["{\text{has as power source}}", from=1-3, to=3-3]
	\arrow[from=1-1, to=3-1]
\end{tikzcd}\end{equation}

\paragraph[Populating an olog]\label{para:tipspopulateolog} The examples in Section \ref{para:classifytoolog} cover some basic principles for populating an olog for applications in a given context:
\begin{enumerate}
  \item Introduce types that correspond to `seed' concepts that are relevant to the given context.  For example, in the case of the olog in  \eqref{fig:olog2}, we can begin with $\lceil$a transport vehicle$\rceil$, $\lceil$a type of power source$\rceil$, $\lceil$gas$\rceil$, and so on.
  \item Introduce aspects that connect different seed concepts.  In  \eqref{fig:olog2}, this means introducing the function $\dashv$has as power source$\vdash$ and the two $\dashv$is$\vdash$ arrows from $\lceil$gas$\rceil$ and $\lceil$human power$\rceil$.
  \item Perform categorical operations such as  fiber products to generate more complicated types in the olog; these operations also come with natural arrows such as projections.  In the case of   \eqref{fig:olog2}, this entails constructing the two fiber products  within it.  These fiber products generate the new types $\lceil$a human-powered vehicle$\rceil$ and $\lceil$a gas-powered vehicle$\rceil$; we can then connect them with the types $\lceil$a bicycle$\rceil$ and $\lceil$a gas-powered passenger car$\rceil$ with respective $\dashv$is$\vdash$ arrows.
\end{enumerate}

\section{Distance between  concepts in an olog}\label{sec:quantifysim-1} 

Once we have an olog that includes the relevant concepts in a particular context or application, we can start using the olog to define distances between pairs of concepts, thereby  quantifying the similarity or `degree of analogy' of different concepts.  The idea is simple: every category has an underlying graph, and there are established methods for defining a notion of distance between vertices in a graph (e.g.\ see \cite{wills2020metrics} and  \cite[D41]{handGT}).  

\begin{defn}[shortest-distance metric on an olog]\label{def:sdm}
Given an olog $O$, we can first consider the underlying graph, which is a directed graph.  Suppose we forget the directions of the arrows and consider the associated undirected graph $G$.  Let $V$ be the set of vertices of $G$, and $A$ the set of edges in $G$.  Recall that a path from one vertex $x$ to another vertex $y$ is a sequence of edges $e_1, \cdots, e_m$, for some positive integer $m$, that begins at $x$ and ends at $y$.  Let us assume $G$ is a connected graph with a finite number vertices and edges.  Then for any function $c : A \to \mathbb{R}_{> 0}$, we can define a new function $d : V \times V \to \mathbb{R}_{\geq 0}$ via the formula
\[
d(x,y) = \min \Bigl\{ \sum_{i=1}^n c(e_i) : \text{ there is a path $e_1, \cdots, e_n$ from $x$ to $y$ in $G$}\Bigr\}.
\]
It is easy to see that the function $d$ satisfies the requirements of a metric on the set $V$, thus giving us a notion of `distance' between vertices of the graph $G$, and hence a notion of distance between types (concepts) in the olog $O$.  If the olog $O$ has an underlying graph that is disconnected, we can simply enlarge the codomain of $d$ to $\mathbb{R}_{\geq 0} \cup \{\infty\}$ and define $d(x,y)$ to be $\infty$ when $x,y$ lie in distinct disconnected components of the graph. 
\end{defn}

\begin{eg}
Consider the olog   \eqref{fig:olog2}.  If we assume the function $c$ in Definition \ref{def:sdm} assigns the value $1$ to every edge in the underlying undirected graph of this olog, then with respect to the shortest-distance metric,  the distance between the  concepts `a bicycle' and `a gas-powered passenger car' would be 4, whereas the distance between the concepts `a human-powered vehicle' and `a gas-powered vehicle' would  be 2.  
\end{eg}

Under the shortest-distance metric, the  distance between two concepts depends on the function $c$, and hence the specific olog in use.  If the olog   \eqref{fig:olog2} contains more types and aspects between $\lceil$a bicycle$\rceil$ and $\lceil$a human-powered vehicle$\rceil$, for instance, and $c$ still assigns the value $1$ to every edge, then the distance between the  concepts `a bicycle' and `a gas-powered passenger car' would be greater than 3.   This is not unreasonable since, even for a person, whether or not two concepts are similar depends on the particular context, and also on the amount of knowledge the person has.

\paragraph[Using ologs to compare relations]\label{sec:comparingrels}  Concepts that seem more intangible at first glance - such as  relations among different entities - may also be represented as types in an olog.  Once two concepts are represented by types in the same olog, we can use  ideas from Section \ref{sec:quantifysim-1}  to define a notion of distance between two relations.

The olog in Figure \ref{fig:olog1} already contains an example: we can represent the relation of ownership between a car and a person as the type $\lceil$a pair $(p,c)$ where person $p$ owns car $c\rceil$.

Between two entities, such a person $p$ and a car $c$, there may be different relations that one can speak of.  For example, if a person $p$ owns a car $c$, then the two entities are related by an `ownership' relation.  If a different person $p'$ owns a different car $c'$, then $p'$ and $c'$ are related by the same relation.  On the other hand, if the person $p$ has access to the car $c'$ (e.g.\ $p$ leases the car $c$ but does not own it), then we can say $p$ and $c'$ are related by an `access' relation, which is different from the `ownership' relation.  In analogical reasoning, it is important to be able to compare different relations between the same entities.  

We describe here a systematic way to construct types in an olog that represent relations between entities.  Suppose $\thicksim$ is a relation between two entities, and we write $x \thicksim y$ to represent `$x$ is related to $y$ via the relation $\thicksim$'.  Then we can always construct the  types and aspects as in  \eqref{fig:olog32}.

\begin{equation}\label{fig:olog32}
\begin{tikzcd}
	& {P_\thicksim: \boxed{\begin{matrix}\text{a pair }(x,y)\text{ where} \\ x \text{ is an entity of kind }K_1\\ y \text{ is an entity of kind }K_2\\ \text{and }x \thicksim y\end{matrix}}} \\
	\\
	& {P: \boxed{\begin{matrix}\text{a pair }(x,y)\text{ where} \\ x \text{ is an entity of kind }K_1\\ y \text{ is an entity of kind }K_2\end{matrix}}} \\
	{\boxed{\begin{matrix} \text{an entity of kind }K_1\end{matrix}}} && {\boxed{\begin{matrix} \text{an entity of kind }K_2\end{matrix}}}
	\arrow["{p_1}"', from=3-2, to=4-1]
	\arrow["{p_2}", from=3-2, to=4-3]
	\arrow["{\text{is}}", from=1-2, to=3-2]
\end{tikzcd}\end{equation}

Here, $p_1$ and $p_2$ refers to projections from a pair $(x,y)$ to its first argument $x$ and second argument $y$, respectively.  We can then use the type $P_\thicksim$ as a representation for the relation $\thicksim$ between an entity of kind $K_1$ and an entity of kind $K_2$.

Using fiber products, we can now construct types that represent various kinds of  relations in an olog and understand how they are related to one another.  


\begin{eg}
Consider the olog in  \eqref{fig:olog10}.  To construct this olog, we first begin with the types $A, B$, and $A'$.  We then construct $D$ as the direct product of $B$ with itself, and define $p_1, p_2$ as the projections onto the first and the second arguments, respectively.   For convenience, given any two types $P, Q$ in this olog, we will write $PQ$ to denote the unique aspect in the olog from $P$ to $Q$.  Then we  construct $C''$ as the fiber product of $C'D$ and $CD$.

Now, if we add the type $F$  to the olog as a proxy for the relation `ownership' between two entities, we can further construct $E$ as the fiber product of $FD$ and $CD$, $E'$ as the fiber product of $FD$ and $C'D$, and $E''$ as the fiber product of $EF$ and $E'F$.  If we add another type $\widetilde{F}$ as a proxy for the relation `has access to' between two entities, we can similarly construct $\widetilde{E}$ as the fiber product of $\widetilde{F}D$ and $CD$, $\widetilde{E}'$ as the fiber product of $\widetilde{F}D$ and $C'D$, and $\widetilde{E}''$ as the fiber product of $\widetilde{E}\widetilde{F}$ and $\widetilde{E}'\widetilde{F}$.

Now, for the underlying undirected graph of the olog in  \eqref{fig:olog10},    if we use a function $c$ that assigns the value $1$ to every edge, then with respect to the shortest-distance metric (Definition \ref{def:sdm}),  the distance between the concepts `owns' and `has access to' would then be equal to 2 (attained by the path $FD$ followed by $D\widetilde{F}$.  The distance between the concept `a person owning a building' (represented by $E''$) and the concept `a person having access to a building' (represented by $\widetilde{E}''$) would also be 2, attained by the path $E''C''$ followed by $C''\widetilde{E}''$.
\end{eg}

One can imagine that, by using other types representing  relations other than $F$ and $\widetilde{F}$, new types representing other relations among different kinds of entities can be added to the olog.  As a result, we will be able to define a distance between any two relations as long as they appear as types in the same olog.

\begin{equation}\label{fig:olog10}
\begin{tikzcd}[scale cd=0.4]
	&&& {\widetilde{E}''\, \, \boxed{\begin{matrix} \text{a pair }(e, f)\text{ where} \\ e \text{ is a person,}\\    f\text{ is a building, and} \\ e\text{ has access to }f\end{matrix}}} &&& {\widetilde{E}'\, \, \boxed{\begin{matrix} \text{a pair }(e, f)\text{ where} \\ e \text{ is an entity,}\\    f\text{ is a building, and} \\ e\text{ has access to }f\end{matrix}}} \\
	& {\widetilde{E} \,\, \boxed{\begin{matrix} \text{a pair }(e, f)\text{ where} \\  e \text{ is a person,} \\ f\text{ is an entity, and} \\ e \text{ has access to }f\end{matrix}}} &&& {\widetilde{F} \,\,\boxed{\begin{matrix} \text{a pair }(e, f)\text{ where} \\  e, f\text{ are entities, and} \\ e \text{ has access to }f\end{matrix}}} \\
	&& {E''\, \, \boxed{\begin{matrix} \text{a pair }(e, f)\text{ where} \\ e \text{ is a person,}\\    f\text{ is a building, and} \\ e\text{ owns }f\end{matrix}}} &&& {E'\, \, \boxed{\begin{matrix} \text{a pair }(e, f)\text{ where} \\ e \text{ is an entity,}\\    f\text{ is a building, and} \\ e\text{ owns }f\end{matrix}}} \\
	{E\,\, \boxed{\begin{matrix} \text{a pair }(e, f)\text{ where} \\  e \text{ is a person,} \\ f\text{ is an entity, and} \\ e \text{ owns }f\end{matrix}}} &&& {F\,\,\boxed{\begin{matrix} \text{a pair }(e, f)\text{ where} \\  e, f\text{ are entities, and} \\ e \text{ owns }f\end{matrix}}} \\
	&& {C''\, \, \boxed{\begin{matrix} \text{a pair }(e, f)\text{ where} \\ e \text{ is a person, and}\\    f\text{ is a building}\end{matrix}}} &&& {C'\, \, \boxed{\begin{matrix} \text{a pair }(e, f)\text{ where} \\ e \text{ is an entity, and}\\    f\text{ is a building}\end{matrix}}} \\
	{C\, \, \boxed{\begin{matrix} \text{a pair }(e, f)\text{ where} \\ e \text{ is a person, and}\\    f\text{ is an entity}\end{matrix}}} &&& {D\,\, \boxed{\begin{matrix} \text{a pair }(e, f)\text{ where} \\  e, f\text{ are entities}\end{matrix}}} && {A'\,\,\fbox{a building}} \\
	{A \, \,  \fbox{a person}} &&& {B\, \, \fbox{an entity}}
	\arrow["{\text{is}}"{pos=0.7}, from=4-4, to=6-4]
	\arrow["{p_1}"', shift right=3, from=6-4, to=7-4]
	\arrow["{\text{is}}", from=7-1, to=7-4]
	\arrow["{p_2}", shift left=3, from=6-4, to=7-4]
	\arrow[from=6-1, to=7-1]
	\arrow["{\text{is}}", from=6-1, to=6-4]
	\arrow["{\text{is}}"{pos=0.2}, from=4-1, to=4-4]
	\arrow[from=4-1, to=6-1]
	\arrow["{\text{is}}", from=6-6, to=7-4]
	\arrow[from=5-6, to=6-6]
	\arrow["{\text{is}}"', from=5-6, to=6-4]
	\arrow[from=3-6, to=5-6]
	\arrow["{\text{is}}"', from=3-6, to=4-4]
	\arrow["{\text{is}}"{pos=0.3}, from=5-3, to=5-6]
	\arrow["{\text{is}}"', from=5-3, to=6-1]
	\arrow["{\text{is}}", from=3-3, to=3-6]
	\arrow["{\text{is}}"', from=3-3, to=4-1]
	\arrow["{\text{is}}"{pos=0.7}, from=3-3, to=5-3]
	\arrow["{\text{is}}"'{pos=0.1}, from=2-5, to=6-4]
	\arrow["{\text{is}}"'{pos=0.1}, from=2-2, to=6-1]
	\arrow["{\text{is}}"'{pos=0.1}, from=1-4, to=5-3]
	\arrow["{\text{is}}"'{pos=0.1}, from=1-7, to=5-6]
	\arrow["{\text{is}}"', from=1-4, to=2-2]
	\arrow["{\text{is}}", from=2-2, to=2-5]
	\arrow["{\text{is}}", from=1-4, to=1-7]
	\arrow["{\text{is}}"{description}, from=1-7, to=2-5]
\end{tikzcd}\end{equation}

\section{Using wiring diagrams to represent processes}\label{sec:WDforprocesses}

In Section \ref{sec:distbtwconcepts-start}, we saw that fiber product from category theory can be used to build complex concepts from simple concepts.  In this section, we introduce the idea of wiring diagrams, which will allow us to represent processes that occur over time.  We will define a wiring diagram as a graph decorated with labels.  Note that the idea of wiring diagrams has been used in engineering for many years, and there have been various mathematical approaches to using  wiring diagrams  to represent systems or processes \cite{SR-WDdisc,spivak2013operad,VSL}.  In this article, we focus on a more elementary approach and leave a more operadic approach as taken in the aforementioned articles  to future work.

\begin{defn}
    A (\emph{directed}) \emph{graph} is a quadruple $G = (V, A, s, t)$ where
    \begin{itemize}
        \item $V$ is a set, the elements of which are called \emph{vertices};
        \item $A$ is a set, the elements of which are called \emph{arrows}, or \emph{edges}; 
        \item $s : A \to V$ is a function, called the \emph{source function};
        \item $t : A \to V$ is a function, called the \emph{target function}.
    \end{itemize}
\end{defn}
  Given an arrow $a$ in a graph, we often draw it as an arrow pointing from $s(a)$ to $t(a)$: 
\[
s(a) \overset{a}{\longrightarrow} t(a).
\]
That is, the functions $s$ and $t$ indicate where an arrow starts and end, respectively.

For example, Figure  \ref{fig:olog12} is a graph with 4 vertices (labelled $A, B, C, D$) and 5 arrows (labelled $a, b, c, d, e$), with the functions $s, t$ given by $s(a)=A=s(b), t(a)=t(b)=s(c)=s(e)=B, t(c)=s(d)=C$, and $t(d)=t(e)=D$.

\begin{figure}[h]
    \centering
    \includegraphics[scale=0.5]{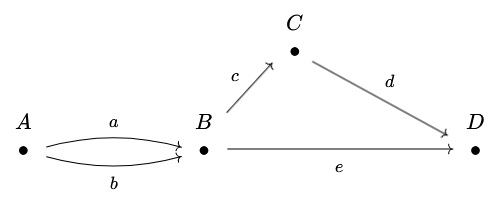}
    \caption{}
    \label{fig:olog12}
\end{figure}

A wiring diagram in this article will be a graph where, to each vertex, we attach a `state vector' that represents the values of certain parameters that correspond to sensors.  To make this precise, we first define the notion of sensing functions.

\begin{defn}[sensing function]
    A \emph{sensing function} $F$ associated to a sensor is a function whose domain $D_F$ is the set of all the things the sensor can be applied to, and whose codomain $C_F$ is the set of all possible outputs from the sensor.  That is, for any $x \in D$, $F(x)$ is the output given by the sensor.
\end{defn}

We  also allow sensors to take on broader meanings, and use the word sensor  to refer to any device or algorithm that observes the environment and gives an output.

\begin{eg} (a) For  a speed sensor next to a motorway,  the corresponding sensing function $F$ can be defined on the set $D_F$ of all cars that move past the sensor, while $C_F$ is the set of all nonnegative integers.  Then, at any point in time, $F$ would give an integer output representing the speed of a vehicle currently in front of the sensor (in miles per hour, rounded to the nearest nonnegative integer); if there is no vehicle in front of the sensor, the sensor gives an output of $0$.

(b) For a motion sensor that detects whether there is any movement inside a particular room,  we can take the domain $D_F$ of  the corresponding sensing function $F$ to be the singleton set $\{\bullet\}$,    and take the codomain $C_F$ of $F$ to be $\{0,1\}$.  This way, at any point  in time, the sensing function $F$ would give the value $0$ if  no movement is detected inside the room, and give the value $1$ if movement is detected.  

(c) For a sensor that tracks the blood oxygen level of a particular person, the corresponding sensing function $F$ can be defined as a function from a singleton set $\{\bullet\}$ to the real interval $[0,100]$ so that, at any point in time, the number $F(\bullet)$ represents the blood oxygen level of that person (in percentage) at that  time.

(d) For any two cities $x, y$ in the world, we can define a sensing function $F_{x,y} : \{\bullet\} \to \{0,1\}$, which depends on $x$ and $y$, such that $F_{x,y}(\bullet)$ takes on the value $1$ (resp.\ $0$) if $x$ and $y$ are sister cities (resp.\ are not sister cities).  The value of $F_{x,y}$ would depend on the particular time when the measurement is taken.  In practice, this sensing function can be constructed using an algorithm that crawls through the public web or government databases to determine the current status of sister city agreements between the two cities.
\end{eg}

We now define a wiring diagram as follows.

\begin{defn}[wiring diagram]\label{def:WD}
    A \emph{wiring diagram} (\emph{WD}) is a quintuple 
    \[
    (V, A, s, t, \mathscr{L}_V)
    \]
    satisfying the following conditions.
    \begin{enumerate}
        \item[WD0.] $G=(V,A,s,t)$ is a finite directed graph, called the \emph{underlying graph} of the wiring diagram.  We will refer to elements of $V$ as \emph{vertices} or \emph{states}, and refer to elements of $A$ as \emph{arrows} or \emph{wires}. 
           
        \item[WD1.] $\mathscr{L}_V$ is an indexed set $\{ L_v\}_{v \in V}$ such that each $L_v$ is a  set of triples
        \[
          L_v = \{ (F_i, x_i, y_i) : 1 \leq i\leq m_v\}
        \]
        where $m_v$ is a nonnegative integer depending on $v$, and where each $F_i$ is a sensing function, with $x_i$ in the domain of $F_i$ and $y_i$ in the codomain of $F_i$. We allow $L_v$ to be the empty set.
        \item[WD2.] There is a labelling of the vertices, given by a function $f : V \to \{1, 2, \cdots, n\}$ where $n$ is the number of elements in $V$, such that for each $a \in A$, we have $f(s(a)) < f(t(a))$.
    \end{enumerate}
\end{defn}
In our formal definition above, we do not \textit{a priori} require $f$ to be a bijective function.   By Lemma \ref{lem:WDgraphsRaxal} below, however, there always exists a bijection $f$ that satisfies  WD2.  By abuse of notation, we will refer to  $L_v$ as the \emph{state vector} at the vertex $v$, and refer to an element of $L_v$ as a \emph{label}.  We will refer to a graph that arises as the underlying graph of a wiring diagram as a \emph{wiring diagram graph} or a \emph{WD graph}.  Note that a finite directed graph is a WD graph if and only if it satisfies WD2.

\medskip
\noindent
\textbf{How to read a wiring diagram.}   
\begin{itemize}
    \item For each vertex $v$, the state vector $L_v$  specifies the values of various parameters that must be achieved at a particular point in time. 
    \item Each arrow $a$ represents the requirement that the state vector $L_{s(a)}$ is achieved \emph{before} the state vector $L_{t(a)}$ is achieved.
\end{itemize} 

Condition WD2  implies that we can always arrange the vertices of a wiring diagram in a way so that every arrow points from left to right.  A wiring diagram then represents a process where, as we read the diagram from left to right, specific readings of parameter values occur.

\paragraph[Properties of  WD graphs]  We list some basic properties of WD graphs in this subsection.

We adopt the following definitions for  a directed graph $G=(V,A,s,t)$.
\begin{itemize}
    \item A loop is an arrow that points from a vertex to itself, i.e.\ $a \in A$ such that $t(a)=s(a)$.
    \item  A path of length $n$ is a sequence of arrows $a_1, \cdots, a_n$, where $n$ is a positive integer, such that $t(a_i)=s(a_{i+1})$ for all $1 \leq i \leq n-1$, and none of the $a_i$ are loops.
    \item  An  oriented cycle is a path of length more than 1 that begins and ends at the same vertex.
\end{itemize}  
Note that an arrow is a path of length 1.

\begin{lem}\label{lem:WDgraph-1}
Let $G=(V,A,s,t)$ be a WD graph.  Then
\begin{itemize}
    \item[(i)] $G$ contains no loops.
    \item[(ii)] $G$ contains no oriented cycles.
\end{itemize}
\end{lem}

\begin{proof}
Let $f$ be a function associated to $G$ in WD2.  

(i) Given any arrow $a \in A$,  we have $f(s(a))<f(t(a))$ which implies $s(a)$ and $t(a)$ must be distinct vertices, and so $a$ cannot be a loop.

(ii) Suppose $G$ contains an oriented cycle, i.e.\ a path formed by the concatenation of arrows $a_1, a_2, \cdots, a_k, a_{k+1}$ where  $t(a_i)=s(a_{i+1})$ for $1 \leq i \leq k$ while $t(a_{k+1})=s(a_1)$.  Then we have 
\[
f(s(a_1))<f(t(a_1))=f(s(a_2))< \cdots  < f(t(a_k))=f(s(a_{k+1}))<f(t(a_{k+1}))=f(s(a_1))
\]
whic is a contradiction.
\end{proof}



\begin{lem}\label{lem:WDgraphsRaxal}
Let $G=(V,A,s,t)$ be a WD graph. Then there exists a bijective function $\widetilde{f} : V \to \{1, 2, \cdots, n\}$ where $n=|V|$ such that for every arrow $a \in A$, we have $\widetilde{f}(s(a))< \widetilde{f}(t(a))$.  
\end{lem}

\begin{proof}
    Let $f : V \to \{1, \cdots, n\}$ be a function associated to $G$ as in WD2.  Without loss of generality, we can assume $\image f = \{1, \cdots, k\}$ for some integer $1 \leq k\leq n$.  For each $i \in \image f$, let $M_i$ denote the preimage of $i$, i.e.\ $M_i = f^{-1}(i)$, and let $m_i = |M_i|$, the number of vertices mapping onto $i$ under $f$. Let $m=\max \{m_i: i \in \image f\}$.

    Now for each $i \in \image f$, let $\widetilde{f}_i$ denote any bijection 
    \[
     M_i \to \left\{ i + \frac{1}{2m}j : 0 \leq j \leq m_i-1\right\} =: R_i.
    \]
    We can then concatenate the $\widetilde{f}_i$ into a single function $\widetilde{f}'$ on $V$, i.e.\ we set
    \begin{align*}
        \widetilde{f}' : V &\to \bigcup_{i\in \image f}R_i \\
        v &\mapsto \widetilde{f}_i(v) \text{\quad if $v \in M_i$}.
    \end{align*}
    Note that $\widetilde{f}'$ is a bijection, and so we can post-compose $\widetilde{f}'$ with a unique order-preserving bijection onto $\{1, \cdots, n\}$ to form a function $\widetilde{f}$.  We claim that $\widetilde{f}$ satisfies WD2.

    Take any arrow $a \in A$.  From the definition of $f$, we have $f(s(a)) < f(t(a))$.  By construction of $\widetilde{f}'$, it follows that \[
    \lfloor \widetilde{f}'(s(a))\rfloor = f(s(a)) < f(t(a)) = \lfloor \widetilde{f}' (t(a)) \rfloor
    \]
    (where $\lfloor - \rfloor$ denotes the floor function on real numbers).  From the construction of $\widetilde{f}'$, this means that $\widetilde{f}'(s(a))$ and $\widetilde{f}'(t(a))$ lie in distinct $R_i$  and   $\widetilde{f}'(s(a)) < \widetilde{f}'(t(a))$, which in turn implies $\widetilde{f}(s(a)) < \widetilde{f}(t(a))$.  That is, $\widetilde{f}$ satisfies WD2.
\end{proof}

\begin{rem}
By Lemma \ref{lem:WDgraphsRaxal}, every WD graph is a directed graph with a linear extension ordering \cite[Section 3.4, D30]{handGT}; as a result, a directed graph is a WD graph (i.e.\ satisfies WD2) if and only if it is a directed acyclic graph (DAG) \cite[Section 3.4, F23]{handGT}.  We will continue to use the term \emph{wiring diagram graph} instead of \emph{directed acyclic graph} in this article, however, to emphasize that we are not merely considering DAGs in this article, but DAGs with extra structures that make them wiring diagrams.
\end{rem}


\paragraph[State vectors and actions]\label{sec:statesvsactions}  Informally, each state vector $L_v$ in a wiring diagram represents the `status' of relevant parameters, whereas each wire $a$ in a wiring diagram represents a `difference of states', and thus corresponds to an action or event that leads the state vector at $L_{s(a)}$ to become the state vector $L_{t(a)}$.  If we define sensing functions carefully, a single state vector can also indicate the occurrence of an action or an event.

\begin{eg}\label{eg:insidecoffeeshop}
For example, suppose we want to represent the concept ``person $p$ enters coffee shop $s$'' using a wiring diagram.  Consider a sensor tracking the movement of the person $p$, where the sensor gives the output $0$ when  $p$ is outside the coffee shop, and gives the output $1$ when $p$ is inside the coffee shop.  This results in a sensing function $F_1$ with the singleton set $\{\bullet\}$ as the domain and defined by
 \[
    F_1(\bullet)= \begin{cases} 0 &\text{ if $p$ is outside $s$} \\
    1 &\text{ if $p$ is inside $s$}\end{cases}
    \]
at any point in time.  The concept ``person $p$ enters coffee shop $s$'' can then be represented by the wiring diagram with two vertices as in  \eqref{fig:olog14}.

\begin{equation}\label{fig:olog14}
\begin{tikzcd}
	{\begin{matrix}\bullet \\ (F_1,\bullet,0)\end{matrix}} && {\begin{matrix}\bullet \\ (F_1,\bullet,1)\end{matrix}}
	\arrow[shift left=2, from=1-1, to=1-3]
\end{tikzcd}\end{equation}

In particular, the single wire in this wiring diagram informally represents the act  of  $p$ `entering' the coffee shop $s$.

Another way to represent the concept ``person $p$ enters coffee shop $s$''
is by considering a `numerical derivative' of $F_1$.  Let us define a new sensing function $dF_1 : \{\bullet\} \to \{-1,0,1\}$ given by 
 \[
    dF_1(\bullet)= (\text{current value of }F_1)-(\text{value of }F_1\text{ five seconds ago}).
    \]
Then the occurrence of $p$ entering the coffee shop $s$ would correspond to the moment when the function $dF_1$ registers a value of $1$, and so the act of $p$ entering the coffee shop $s$ can also be represented by the following wiring diagram with a single vertex and no wires

\begin{equation*} \label{fig:olog20}
\begin{tikzcd}
	{\begin{matrix}\bullet \\ (dF_1,\bullet,1)\end{matrix}}
 \end{tikzcd}\end{equation*}

\end{eg}

\begin{eg}\label{eg:WDgraphsonsameV}
In \eqref{fig:olog17} there are  three possible underlying graphs for wiring diagrams with four vertices.  In a wiring diagram with underlying graph $G_1$, the state vectors must be achieved in the order of $L_A, L_B, L_C$, and then $L_D$. A situation where such a wiring diagram arises, for example, would be in curriculum planning.  In the curriculum of a gentle introduction to calculus, for example, we can define the state vectors $L_A$ through $L_D$ to represent the following:
\begin{itemize}
    \item $L_A$: A student has learned the definition of continuity.
    \item $L_B$: A student has learned to take limits of functions.
    \item $L_C$: A student has learned the definition of derivative.
    \item $L_D$: A  student has learned to take the derivative of a polynomial function.
\end{itemize}
To formally write these as labels of a wiring diagram in terms of sensing functions, one can use sensing functions that register the scores on tests covering the respective topics.

\begin{equation}\label{fig:olog17}
\begin{tikzcd}[scale cd=0.8]
&&&&&&&& {\begin{matrix} \bullet \\ B \end{matrix}} \\
	{G_1:} & {\begin{matrix} \bullet \\ A \end{matrix}} & {\begin{matrix} \bullet \\ B \end{matrix}} & {\begin{matrix} \bullet \\ C \end{matrix}} & {\begin{matrix} \bullet \\ D \end{matrix}} && {G_2:} & {\begin{matrix} \bullet \\ A \end{matrix}} && {\begin{matrix} \bullet \\ D \end{matrix}} \\
	&&&&&&&& {\begin{matrix} \bullet \\ C\end{matrix}} \\
	&&& {G_3:} & {\begin{matrix} \bullet \\ A \end{matrix}} && {\begin{matrix} \bullet \\ B \end{matrix}} \\
	&&&& {\begin{matrix} \bullet \\ C \end{matrix}} && {\begin{matrix} \bullet \\ D \end{matrix}}
	\arrow[shift left=3, from=2-2, to=2-3]
	\arrow[shift left=3, from=2-3, to=2-4]
	\arrow[shift left=3, from=2-4, to=2-5]
	\arrow[shift left=2, from=2-8, to=1-9]
	\arrow[shift left=2, from=1-9, to=2-10]
	\arrow[shift left=2, from=2-8, to=3-9]
	\arrow[shift left=2, from=3-9, to=2-10]
	\arrow[shift left=2, from=4-5, to=4-7]
	\arrow[shift left=2, from=4-5, to=5-7]
	\arrow[shift left=2, from=5-5, to=5-7]
\end{tikzcd}\end{equation}

In a wiring diagram with underlying graph $G_2$, the state vector $L_A$ must be achieved first; then $L_B$ and $L_C$ must be achieved, but between them it does not matter whether it is $L_B$ or $L_C$ that is achieved first.  After both $L_B, L_C$ have been achieved, $L_D$ must then be achieved.  Our example in Section \ref{eg:buyingcoffee}  involves a wiring diagram that contains $G_2$ as part of its underlying graph.

In the case of $G_3$, either $L_A$ or $L_C$ must be achieved first, although they can occur independently.  The state vector $L_D$ can only be achieved after $L_A, L_C$ have both been achieved, while $L_B$ can only be achieved after $L_A$ has been achieved.  (Between $L_B$ and $L_D$, there is no requirement as to which should come first.)  A situation where such a wiring diagram  arises is when different teams work on a project continuously in relay.  Suppose for any positive integer $i$, there is a team $T_i$, and that all these teams work on the same project in a factory.  Then a wiring diagram with underlying graph $G_3$ would depict the process of `passing the baton' from one team to the next if we take the state vectors to represent the following for any $i\geq 2$:
\begin{itemize}
    \item $L_A$: $T_{i-1}$ has created instructions for $T_i$.
    \item $L_B$: $T_{i-1}$ has departed the factory.
    \item $L_C$: $T_i$ has arrived at the factory.
    \item $L_D$: $T_i$ has completed  the instructions from $T_{i-1}$ and made new progress on the project.
\end{itemize}
A wiring diagram with underlying graph such as  \eqref{fig:olog18} would then represent the entire relay process.

\begin{equation}\label{fig:olog18}
\begin{tikzcd}
	\bullet & \bullet \\
	\bullet & \bullet & \bullet & \bullet \\
	&& \bullet & \bullet & \bullet & \bullet \\
	&&&& \bullet & \bullet \\
	&&&&&& \ddots & \bullet & \bullet \\
	&&&&&&& \bullet & \bullet
	\arrow[from=1-1, to=1-2]
	\arrow[from=1-1, to=2-2]
	\arrow[from=2-1, to=2-2]
	\arrow[from=2-2, to=2-3]
	\arrow[from=2-3, to=2-4]
	\arrow[from=2-3, to=3-4]
	\arrow[from=3-3, to=3-4]
	\arrow[from=3-4, to=3-5]
	\arrow[from=3-5, to=3-6]
	\arrow[from=3-5, to=4-6]
	\arrow[from=4-5, to=4-6]
	\arrow[from=5-8, to=5-9]
	\arrow[from=5-8, to=6-9]
	\arrow[from=6-8, to=6-9]
\end{tikzcd}\end{equation}
\end{eg}

\paragraph[Example: buying coffee]\label{eg:buyingcoffee}  Let us build on Example \ref{eg:insidecoffeeshop} and consider the process of ``a person $p$ buying coffee from a shop $s$''.  We can think of this process as comprising  four components:
\begin{itemize}
    \item[(i)] $p$ enters the coffee shop $s$.
    \item[(ii)] $p$ makes payment for coffee.
    \item[(iii)] $p$ receives coffee.
    \item[(iv)] $p$ leaves the coffee shop $s$.
\end{itemize}
Depending on the type of shop, (ii) might occur before (iii), or (iii) might occur before (ii); it is reasonable to assume, however, that in most cases, (i) must occur before both (ii) and (iii), which must occur before (iv).

Next, we describe each of events (i) through (iv) in terms of   sensors.  We can use changes in the value of the sensing function $F_1$ from above to describe (i) and (iv).  To describe the event (ii), consider a sensing function $F_2$  that detects whether a payment has been made by $p$ for coffee (this is something that can be detected as a change in the activity log in the cashier's machine, or in the activity log of person $p$'s payment devices).     That is, we can take $F_2$ to be a function with domain $\{\bullet\}$ such that
    \[
    F_2(\bullet)= \begin{cases} 0 &\text{ if $p$ has not made a new payment for coffee} \\
    1 &\text{ if $p$ has  made a new payment for coffee}\end{cases}.
    \]

To describe the event (iii), we can define a sensing function $F_3$ that detects whether $p$ is holding coffee (such as from an image recognition algorithm), i.e.\   $F_3$ has domain $\{\bullet\}$ and is given by 
    \[
    F_3(\bullet)= \begin{cases} 0 &\text{ if $p$ is not holding any coffee} \\
    1 &\text{ if $p$ is holding coffee}\end{cases}.
    \]

The process of ``a person $p$ buying coffee from a coffee shop $s$'' can now be represented by the wiring diagram in  \eqref{fig:olog15}.

\begin{equation}\label{fig:olog15}
\begin{tikzcd}
	&&&& {\begin{matrix} \bullet \\  (F_2,\bullet,1)  \\ \text{}\\\text{} \end{matrix}} \\
	{\begin{matrix} \bullet \\ (F_1,\bullet,0) \\ (F_2,\bullet,0) \\ (F_3,\bullet,0)  \end{matrix}} && {\begin{matrix} \bullet \\ (F_1,\bullet,1) \\ \text{}\\\text{}  \end{matrix}} &&&& {\begin{matrix} \bullet \\ (F_1,\bullet,0)\\ \text{}\\\text{}  \end{matrix}} \\
	&&&& {\begin{matrix} \bullet \\  (F_3,\bullet,1) \\ \text{}\\\text{}  \end{matrix}}
	\arrow[shift left=5, from=2-1, to=2-3]
	\arrow[shift left=5, from=2-3, to=1-5]
	\arrow[shift left=5, from=2-3, to=3-5]
	\arrow[shift left=5, from=1-5, to=2-7]
	\arrow[shift left=5, from=3-5, to=2-7]
\end{tikzcd}\end{equation}

Let us define numerical derivatives of $F_2$ and $F_3$ by setting
 \[
    dF_i(\bullet)= (\text{current value of }F_i)-(\text{value of }F_i\text{ five seconds ago}).
    \]
for $i=1,2,3$.  Then the process of ``a person $p$ buying coffee from a coffee shop $s$'' can also be represented by the wiring diagram 

\begin{equation}\label{fig:olog21}
\begin{tikzcd}
	&& {\begin{matrix} \bullet \\ (dF_2,\bullet,1)  \end{matrix}} \\
	{\begin{matrix} \bullet \\ (dF_1,\bullet,1)  \end{matrix}} &&&& {\begin{matrix} \bullet \\ (dF_1,\bullet,-1)  \end{matrix}} \\
	&& {\begin{matrix} \bullet \\ (dF_3,\bullet,1)  \end{matrix}}
	\arrow[from=2-1, to=1-3]
	\arrow[from=1-3, to=2-5]
	\arrow[from=2-1, to=3-3]
	\arrow[from=3-3, to=2-5]
\end{tikzcd}\end{equation}

In this wiring diagram, each of the four labels corresponds to an action by $p$.

\paragraph[Wiring diagrams and ologs]\label{sec:WDandologs} Recall that every label in a wiring diagram is of the form $(F,x,y)$ where $F$ is a sensing function with some domain $D_F$ and codomain $C_F$.  Fix an element $y_0$ of $C_F$.  We can construct the olog   \eqref{fig:olog19}, where each vertical square is constructed using a fiber product.  The instances of the type $\lceil$an element $x$ of $D_F$, $F(x)=y_0\rceil$ correspond to labels of the form $(F,x,y_0)$, and so we can take this type in the olog as a representation of  the concept captured by the label $(F,x,y)$.  This way, every label in a wiring diagram can be represented by a type in an olog.  Since we can define the distance between any two types in an olog (Section \ref{sec:quantifysim-1}), we can define the distance between any two labels in a wiring diagram once we represent them as types in the same olog.

\begin{equation}\label{fig:olog19}
\begin{tikzcd}
	& {\boxed{\begin{matrix}\text{an element }x\text{ of }D_F, \\ F(x)=y_0\end{matrix}}} && {\boxed{\begin{matrix}\text{an element of }D_F \end{matrix}}} \\
	{\boxed{\begin{matrix}\text{an element }x\text{ of }D_F, \\ F(x)\neq y_0\end{matrix}}} \\
	& {\{y_0\}} && {\boxed{\begin{matrix}\text{an element of }C_F \end{matrix}}} \\
	{\boxed{\begin{matrix}\text{an element }y\text{ of }C_F, \\ y\neq y_0\end{matrix}}}
	\arrow["F", from=1-4, to=3-4]
	\arrow["{\text{is}}", from=3-2, to=3-4]
	\arrow["{\text{is}}", from=1-2, to=1-4]
	\arrow[from=1-2, to=3-2]
	\arrow["{\text{is}}"', from=2-1, to=1-4]
	\arrow["{\text{is}}"', from=4-1, to=3-4]
	\arrow[from=2-1, to=4-1]
\end{tikzcd}\end{equation}

\paragraph[Relations and sensing functions]\label{sec:relnsfs}  Even though a label in a wiring diagram must be of the form $(F,x,y)$ by definition, this definition is broad enough to describe relations among entities.  To see this, let us recall some basic terminology on relations on sets.  

Given two sets $S$ and $T$, a \emph{binary relation} $R$ on $S \times T$ is simply defined to be a subset of $S \times T$.  For any $a \in S$ and $b \in T$, we say \emph{$a$ is related to $b$} or write $a\thicksim b$ (when $R$ is understood) to mean $(a,b) \in R$.

Given a set $S$, a binary relation $R$ on $S$ is  defined to be a relation on $S \times S$.  For a relation $R$ on a set $S$, we say $R$ is
\begin{itemize}
    \item reflexive if   $x \thicksim x$ for all $x \in S$;
    \item anti-symmetric if, whenever $x \thicksim y$ and $y \thicksim x$ for $x, y \in S$, it follows that $x=y$;
    \item transitive if, whenever $x \thicksim y$ and $y \thicksim z$ for $x, y, z \in S$, we have $x \thicksim z$.
\end{itemize}
A binary relation that is both reflexive and transitive is called a \emph{preorder}; a preorder that is also anti-symmetric is called a \emph{partial order}.

Given a relation $R$ on a set $A \times B$, we can  define the function 
\[
  q_R : A \times B \to \{0,1\} : (a,b) \mapsto \begin{cases} 0 &\text{ if $a \nsim b$} \\
  1 &\text{if $a \thicksim b$} \end{cases}.
\]
That is, $F$ is a function that detects whether a pair $(a,b)$ satisfies the relation $R$.  If we write $i_j$ to denote the inclusion of $\{j\}$ into $\{0,1\}$ for $j=0,1$, then we can construct the  olog  \eqref{fig:olog27} where each vertical square is a fiber product.  Pairs $(a,b)$ that lie in $R$ are now instances of the type $P_1$, and so we can take $P_1$ as a type that represents the relation $R$.

\begin{equation}\label{fig:olog27}
\begin{tikzcd}
	& {P_0: \boxed{\begin{matrix}\text{a pair }(a,b)\text{ where} \\ a\in A, b \in B \\ \text{and }a \nsim b \end{matrix}}} \\
	{P_1: \boxed{\begin{matrix}\text{a pair }(a,b)\text{ where} \\ a\in A, b \in B \\ \text{and }a \thicksim b \end{matrix}}} &&& {P: \boxed{\begin{matrix}\text{a pair }(a,b)\text{ where} \\ a\in A, b \in B \end{matrix}}} \\
	& {\{0\}} \\
	{\{1\}} &&& {\{0,1\}}
	\arrow["{q_R}", from=2-4, to=4-4]
	\arrow["{i_0}", from=3-2, to=4-4]
	\arrow["{i_1}"', from=4-1, to=4-4]
	\arrow[from=1-2, to=3-2]
	\arrow[from=1-2, to=2-4]
	\arrow[from=2-1, to=4-1]
	\arrow[from=2-1, to=2-4]
\end{tikzcd}\end{equation}

Note that we can regard $q_R$ as a sensing function with domain $A \times B$ and codomain $\{0,1\}$; this way, the olog   \eqref{fig:olog27} is merely a special case of the olog   \eqref{fig:olog19}.  The concept ``$a$ is related to $b$ with respect to the relation $R$'' can now be represented by the label $(q_R, (a,b), 1)$ in a wiring diagram.  Equivalently, we can rewrite the label as $(F,\bullet, 1)$ where $F$ is the sensing function 
\[
F_{(a,b),R} : \{\bullet \} \to \{0,1\} : \bullet \mapsto \begin{cases} 0 &\text{ if $a \nsim b$} \\ 1 &\text{ if $a \thicksim b$} \end{cases}.
\]

\begin{rem}
In the previous section, we saw that every label in a wiring diagram can be represented by  a type in an olog.  In the current section, we saw that labels can represent whether two entities satisfy a relation.
\end{rem}

\section{Using wiring diagrams to quantify analogy}\label{sec:catofWDs}

In Section \ref{sec:quantifysim-1}, we saw that there is a way to define the distance between any two concepts that occur as types in the same olog.  In Section \ref{sec:comparingrels}, we saw that the relations between different entities (such as ownership or access) can be represented as types in an olog.  Then, in Section \ref{sec:WDforprocesses}, we saw that wiring diagrams can represent processes that occur over a period of time, and that the labels at the vertices of a wiring diagram can be defined using concepts that occur in an olog.  This allowed us to conclude in Section \ref{sec:WDandologs} that we can define the distance between any two labels in a wiring diagram, as long as they both correspond to types in the same olog.

In this section, we propose a definition of  distance between any two wiring diagrams.  Our definition builds on the idea of elementary edit operations between graphs - which leads to the notion of graph edit distance -  taking advantage of the fact that every wiring diagram has an underlying graph.  Our approach has two advantages compared to simply considering the graph edit distance between the underlying graphs, however.  First, we consider categories generated by these graphs, which allow us to make better use of the inherent structures of wiring diagrams; second, since labels of wiring diagrams correspond to types in an olog, we also have a measure of distance among the labels themselves that takes into account the structure of the olog being used.  That is, our definition refines graph edit distance by utilizing  the categorical aspects of wiring diagrams.

\paragraph[A category of skeleton WD graphs]\label{para:catskWDgs}  All the wiring diagrams that have appeared in this article are `skeleton' in the following sense: 

\begin{defn}
    We say a WD graph $G=(V,A,s,t)$ is \emph{skeleton} if it satisfies the following condition:
    \begin{itemize}
        \item[WD3.] For any two distinct vertices $v, v' \in V$, if there is already a path from $v$ to $v'$ given by arrows $a_1, \cdots, a_n$ in this order, then there cannot be any arrow $a^\ast$ from $v$ to $v'$ such that $a^\ast \neq a_i$ for all $1 \leq i \leq n$.   
    \end{itemize}
We say a wiring diagram is skeleton if its underlying graph is skeleton.
\end{defn}

Note that  for any two distinct vertices $v, v'$ in a skeleton WD graph, there is at most one arrow from $v$ to $v'$.

We now describe a construction that takes any skeleton WD graph  and produces a partial order on its set of vertices.  Suppose $G=(V,A,s,t)$ is a skeleton WD graph.    First, we define a relation $R_0$ on $V$ by setting
\[
R_0 = \{ (x,y) \in V \times V : x=s(a), y=t(a) \text{ for some }a \in A\}.
\]
Next, we define the transitive closure $R_2$ of $R_0$  \cite[D26, Chap.\ 3]{handGT}.  That is, we first define 
\[
R_1 = R_0 \bigcup \{ (x,y) \in V \times V : x=y\},
\]
and then declare an element $(x,y)$ of $V \times V$ to be in $R_2$ if and only if there is a sequence of elements $(x_0,x_1), (x_1,x_2), \cdots, (x_{k-1},x_k)$ in $R_1$ with $k\geq 1$ such that $x_0=x$ and $x_k=y$.  In other words, $R_2$ is the result of forcing reflexivity and transitivity on $R_0$.  By construction, $R_2$ is a preorder on $V$ and $R_0 \subseteq R_2$.  We write $R(G)$ to denote $R_2$.  Note that $R(G)$ depends only on $G$, and not a choice of the bijection $\wt{f}$ from Lemma  \ref{lem:WDgraphsRaxal}.

\begin{lem}\label{lem:RofG-1}
Let $G$ be a skeleton WD graph, and $\wt{f} : V \to \{1, \cdots, n\}$ (where $n=|V|$) any bijection as in Lemma  \ref{lem:WDgraphsRaxal}.  Then 
\begin{itemize}
    \item[(i)] For any element $(x,y)$ of $R(G)$ where $x \neq y$, we have $\wt{f}(x) < \wt{f}(y)$.
    \item[(ii)] If we identify $V$ with the set $\{1, \cdots, |V|\}$ via the bijection  $\wt{f}$, then $R(G)$ is a subset of the natural preorder on $\mathbb{Z}$.
\end{itemize}
\end{lem}

\begin{proof}
(i) Take any $(x,y) \in R(G)$ such that $x \neq y$.  From the construction above, there exists a sequence $(x_0,x_1), (x_1,x_2), \cdots, (x_{k-1},x_k)$ in $R_1$ with $k\geq 1$ such that $x_0=x, x_k=y$.  For each $0 \leq i \leq k-1$, either $x_i=x_{i+1}$ in which case $\wt{f}(x_i)=\wt{f}(x_{i+1})$, or $x_i \neq x_{i+1}$ in which case $(x_i, x_{i+1}) \in R_0$ and $\wt{f}(x_i)<\wt{f}(x_{i+1})$.  The claim then follows.

(ii) This follows immediately from (i).
\end{proof}

From the construction of $R(G)$, it is clear that $R(G)$ is always a partial order on $V$.  Lemma \ref{lem:RofG-1}(ii) says that there is always an embedding of partial orders from $R(G)$ into the natural partial order $(\mathbb{Z}, \leq)$.

For any finite set $V$, we can now define a category $\mathcal{R}(V)$.  We will take the objects of $\mathcal{R}(V)$ to be skeleton WD graphs $G$ whose underlying set of vertices is exactly $V$.  Given any two skeleton WD graphs $G_1, G_2$, we will define a morphism $G_1 \to G_2$ whenever $R(G_2) \subseteq R(G_1)$.  (Note that $R(G_1), R(G_2)$ are both subsets of $V \times V$.)  More precisely, if we consider the category $\mathcal{P}(V \times V)$ of subsets of $V \times V$ where morphisms are set inclusions, then we declare a morphism $\alpha_i : G_1 \to G_2$ in $\mathcal{R}(V)$ whenever there is a morphism $i : R(G_2) \to R(G_1)$ in $\mathcal{P}(V \times V)$.  In particular, for any object $G$ in $\mathcal{R}(V)$, we declare the identity morphism on $G$ to be that corresponding to the identity function on $R(G)$.  Given any two composable morphisms, say $\alpha_i : G_1 \to G_2$ and $\alpha_j : G_2 \to G_3$, we define the composition $\alpha_j \alpha_i$ to be $\alpha_{ij}$, i.e.\ the morphism corresponding to the composite set inclusion $ij : R(G_3) \subseteq R(G_1)$.  It is easy to see that $\mathcal{R}(V)$ satisfies the axioms of a category.  We will refer to $\mathcal{R}(V)$ as the category of skeleton WD graphs over $V$.

Sometimes, we will use $\Rightarrow$ to indicate a morphism in $\mathcal{R}(V)$ to better distinguish between the arrows within wiring diagrams themselves.  We say a morphism $\alpha: G_1 \to G_2$ in $\mathcal{R}(V)$ is \emph{irreducible} if it cannot be written as the composition of two non-identity morphisms, i.e.\ if there is no skeleton WD graph $G_3$ such that $R(G_2) \subsetneq R(G_3) \subsetneq R(G_1)$.

\begin{eg}
Let $V$ be the set $\{A, B, C\}$.  Then $\alpha, \alpha'$ below are morphisms in the category $\mathcal{R}(V)$

\begin{equation}\label{fig:olog24}
\begin{tikzcd}
	&& A & B & C \\
	\\
	A &&&&&& B \\
	& C &&&& A \\
	B &&&&&& C
	\arrow[""{name=0, anchor=center, inner sep=0}, from=1-3, to=1-4]
	\arrow[""{name=1, anchor=center, inner sep=0}, from=1-4, to=1-5]
	\arrow[""{name=2, anchor=center, inner sep=0}, from=3-1, to=4-2]
	\arrow[from=5-1, to=4-2]
	\arrow[""{name=3, anchor=center, inner sep=0}, from=4-6, to=3-7]
	\arrow[from=4-6, to=5-7]
	\arrow["\alpha"', shorten <=15pt, shorten >=15pt, Rightarrow, from=0, to=2]
	\arrow["{\alpha'}", shorten <=15pt, shorten >=15pt, Rightarrow, from=1, to=3]
\end{tikzcd}\end{equation}

The morphism $\alpha$ corresponds to the  inclusion 
\[
\{(A,C), (B,C)\} \subseteq \{ (A,B), (B,C), (A,C)\}
\]
of subsets of $V \times V$ while the morphism $\alpha'$ corresponds to the inclusion 
\[
\{(A,B), (A,C)\} \subseteq \{ (A,B), (B,C), (A,C)\}.
\]
\end{eg}

Informally, having a morphism $G\to G'$ in a category $\mathcal{R}(V)$ means that the partial order generated by $G'$ is `more general' (i.e.\ is a smaller subset of $V \times V$, and hence has `less restrictions') than that generated by $G$.

\paragraph[Morphisms in the category of skeleton WD graphs]  Morphisms in the category of skeleton WD graphs give us a way to compare intrinsic structures of wiring diagrams.  

Let us return to the example  in Section \ref{eg:buyingcoffee}, where we defined sensing functions $dF_1, dF_2, dF_3$ and used them to write down a wiring diagram as in  \eqref{fig:olog21} to represent the process of ``person $p$ buying coffee from coffee shop $s$''.  Different people might have come up with  different wiring diagrams to represent the same process.  Both wiring diagrams in  \eqref{fig:olog22}  can  represent the process of $p$ buying coffee from $s$, the difference being whether we require the person to pay for coffee before or after they receive it.

\begin{equation}\label{fig:olog22}
\begin{tikzcd}
	{\begin{matrix}\bullet\\ (dF_1,\bullet,1)\end{matrix}} & {\begin{matrix}\bullet\\ (dF_2,\bullet,1)\end{matrix}} & {\begin{matrix}\bullet\\ (dF_3,\bullet,1)\end{matrix}} & {\begin{matrix}\bullet\\ (dF_1,\bullet,-1)\end{matrix}} \\
	\\
	{\begin{matrix}\bullet\\ (dF_1,\bullet,1)\end{matrix}} & {\begin{matrix}\bullet\\ (dF_3,\bullet,1)\end{matrix}} & {\begin{matrix}\bullet\\ (dF_2,\bullet,1)\end{matrix}} & {\begin{matrix}\bullet\\ (dF_1,\bullet,-1)\end{matrix}}
	\arrow[shift left=2, from=1-1, to=1-2]
	\arrow[shift left=2, from=1-2, to=1-3]
	\arrow[shift left=2, from=1-3, to=1-4]
	\arrow[shift left=2, from=3-1, to=3-2]
	\arrow[shift left=2, from=3-2, to=3-3]
	\arrow[shift left=2, from=3-3, to=3-4]
\end{tikzcd}\end{equation}

In practice, we would want to consider the two wiring diagrams in  \eqref{fig:olog22} as very `similar' to the wiring diagram in  \eqref{fig:olog21}; in fact, we would normally think of the two diagrams in  \eqref{fig:olog22} as special cases of that in  \eqref{fig:olog21}.  To make these comparisons mathematically precise, we can use the category of skeleton WD graphs.

For simplicity, let us write $A, B, C, D$ to denote the  labels 
\[
(dF_1,\bullet,1),\, \, \, (dF_2,\bullet,1), \, \, \,(dF_3,\bullet,1), \, \, \,(dF_1,\bullet,-1),
\]
respectively.  Let $V$ be the set $\{ A, B, C, D\}$, and consider the category $\mathcal{R}(V)$ of skeleton WD graphs over $V$.  Then we have   two morphisms $\alpha, \alpha'$ in $\mathcal{R}(V)$   as in  \eqref{fig:olog23}.

\begin{equation}\label{fig:olog23}
\begin{tikzcd}
	A & B & C & D &&&& A & C & B & D \\
	\\
	&&&&& B \\
	&&&& A && D \\
	&&&&& C
	\arrow[from=1-1, to=1-2]
	\arrow[""{name=0, anchor=center, inner sep=0}, from=1-2, to=1-3]
	\arrow[from=1-3, to=1-4]
	\arrow[from=1-8, to=1-9]
	\arrow[""{name=1, anchor=center, inner sep=0}, from=1-9, to=1-10]
	\arrow[from=1-10, to=1-11]
	\arrow[""{name=2, anchor=center, inner sep=0}, from=4-5, to=3-6]
	\arrow[""{name=3, anchor=center, inner sep=0}, from=3-6, to=4-7]
	\arrow[from=4-5, to=5-6]
	\arrow[from=5-6, to=4-7]
	\arrow["\alpha"', shorten <=20pt, shorten >=20pt, Rightarrow, from=0, to=2]
	\arrow["{\alpha'}", shorten <=20pt, shorten >=20pt, Rightarrow, from=1, to=3]
\end{tikzcd}\end{equation}

The morphisms $\alpha, \alpha'$ can be considered as mathematical formulations of the similarities between these wiring diagrams. 

\begin{rem}\label{rem:oldvsnewWDmor}
Even though there is a notion of a category of graphs, the morphisms $\alpha, \alpha'$ cannot have been defined as morphisms of graphs.  In fact, using a standard definition of a morphism between graphs \cite[Section II.7]{maclane:71}, in \eqref{fig:olog23} any morphism in the category of graphs from the upper left graph to the lower graph  should take the arrow from $B$ to $C$ to some arrow from $B$ to $C$, whereas we do not have any arrow between $B$ and $C$ in the lower graph.
\end{rem}

\paragraph[Edit distance for graphs]\label{sec:distanceedit-G} For undirected graphs, a standard method for measuring the similarity between graphs is to use the graph edit distance.  To define  graph edit distance, one first needs to decide on a set of elementary edit operations on graphs such as inserting or deleting a vertex, inserting or deleting an edge, or changing the label of a vertex or an edge.  The graph edit distance between two graphs $G, G'$ is then the minimum number of elementary operations needed in order to transform $G$ to $G'$ (e.g.\ see \cite[Section 3.1]{gapGEDKM} or \cite{Riesen:1413749,Serra1,GXTL}).  The graph edit distance is a metric on the set of all finite graphs.  

\paragraph[Distance for wiring diagrams]\label{sec:distanceedit-WD} Since wiring diagrams can be considered as directed graphs where the vertices are labelled with state vectors, we can also define a version of the graph edit distance tailored to wiring diagrams.  In fact, we will introduce new operations on graphs that are only possible by considering the intrinsic structures of wiring diagrams.

Let us write $W_s^\bullet$ to denote the set of all  skeleton wiring diagrams where the state vector $L_v$ at every vertex $v$ is nonempty, and where the underlying graph has at least one vertex.  To begin with, we define \emph{elementary edit operations} on such wiring diagrams to be the following. 
\begin{itemize}
    \item[(i)] Adding a new vertex with a nonempty state vector.
    \item[(ii)] Deleting a vertex along with its state vector.
    \item[(iii)] Adding a new label at a vertex.
    \item[(iv)] Deleting an existing  label at a vertex.
    \item[(v)]  Changing an existing label at a vertex to a different label.
    \item[(vi)] Adding an arrow.
    \item[(vii)] Deleting an arrow.
    \item[(viii)] Replacing the underlying graph $G=(V,A,s,t)$ of a skeleton wiring diagram with  another skeleton WD graph $G'=(V',A',s',t')$, such that $V'=V$ and there is an irreducible morphism $G \to G'$ in $\mathcal{R}(V)$.  
    \item[(ix)] Replacing the underlying graph $G=(V,A,s,t)$ of a skeleton wiring diagram with a another skeleton WD graph $G^\ast=(V^\ast,A^\ast,s^\ast,t^\ast)$, such that $V^\ast=V$ and there is an irreducible morphism $G^\ast \to G$ in $\mathcal{R}(V)$.
\end{itemize}
We require an elementary operation to take a wiring diagram in $W_s^\bullet$ to another wiring diagram in $W_s^\bullet$.  For example, we cannot apply operation (iv) to a vertex if it results in the vertex having an empty state vector, while  operation (vi) is only valid if condition WD2 continues to hold.  We will write $\mathrm{EEO}(W_s^\bullet)$ to represent the set of all possible elementary edit operations on $W_s^\bullet$.



Note that operations (viii) and (ix) only change the arrows in the underlying graph and do not change the state vectors at the vertices.  As we saw in Remark \ref{rem:oldvsnewWDmor}, operations (viii) and (ix) cannot always be replaced with elementary edit operations of other types.     Also, operations of types (i), (iii), (vi),  (viii) are the inverses of operations of types (ii), (iv), (vii), (ix), respectively, while the inverse of an operation of type (v) is again of type (v).

If there is a sequence $E_1, \cdots, E_m$ of elementary edit operations (where $m$ is a positive integer) that transforms a wiring diagram $W$ in $W_s^\bullet$ to another wiring diagram $W'$ in $W_s^\bullet$, then we say $(E_1,\cdots, E_m)$ is an \emph{edit path} from $W$ to $W'$.  Given two wiring diagrams $W, W'$ in $W_s^\bullet$, we will write $P(W,W')$ to denote the set of all edit paths $(E_1, \cdots, E_m)$ that transform $W$ to $W'$.  Note that the length $m$ of the edit path may be different for different paths.

\begin{lem}\label{lem:def:distWD}
Let $c$ be any function from $\mathrm{EEO}(W_s^\bullet)$ to $\mathbb{R}_{>0}$.  For any $W, W' \in W_s^\bullet$, set 
\[
 d(W,W') = \min{\Bigl\{ \sum_{i=1}^m c(E_i) : (E_1, \cdots, E_m) \in P(W,W')    \Bigr\}}.
\]
Then $d$ is a function $W_s^\bullet \times W_s^\bullet \to \mathbb{R}_{>0}$ that defines a metric on $W_s^\bullet$.
\end{lem}

\begin{proof}
Given any wiring diagram $W$ in $W_s^\bullet$, there is always a sequence of elementary edit operations of types (vii), (ii) and (iv) that transform $W$ into a wiring diagram with a single vertex and a single label.  This means that for any two elements $W, W'$ in $W_s^\bullet$, there is always an edit path of finite length from $W$ to $W'$.  Hence $d(W,W')$ is a positive real number, i.e.\ $d$ defines a function from  $W_s^\bullet \times W_s^\bullet$ to $\mathbb{R}_{>0}$.

If we formally define $d(W,W)=0$ for any $W \in W_s^\bullet$, then a standard argument shows that $d$ satisfies the requirements of a metric.
\end{proof}

Note that the distance $d(W,W')$ between two wiring diagrams $W, W'$ in $W_s^\bullet$ depends on two things:
\begin{itemize}
    \item The types of elementary edit operations allowed.
    \item The `cost function' $c$ in  Lemma \ref{lem:def:distWD}.
\end{itemize}
In particular, the cost function $c$ can be designed so as to reflect the olog that represents the internal knowledge of an autonomous system, as the next example shows.

\begin{eg}\label{eg:metricWsbwithologdist}
Fix an olog $O$.  (In practice, $O$ would contain the `internal knowledge' of an autonomous system.)  Assume that all the labels in wiring diagrams that will arise are uniquely represented by types in $O$.  In other words, if we let $\widetilde{L}$ represent the set of all labels that will appear in wiring diagrams considered, and let $T$ denote the set of all the types in $O$, then there is an injection $i : \widetilde{L} \to T$.  Suppose we want to compute $d(W,W')$ for some $W, W' \in W_s^\bullet$.  Let $A$ denote the set of edges in the underlying undirected graph of $O$ (i.e.\ we consider the underlying directed graph of $G$, and then ignore the directions of the arrows).  For any function $c_O : A \to \mathbb{R}_{>0}$, we can define a metric $d_O$ on $T$ as in Definition \ref{def:sdm}.  Now let $c$ be any function from $\mathrm{EEO}(W_s^\bullet)$ to $\mathbb{R}_{>0}$ such that, for any elementary edit operation $E$ of type (v) that changes a label $L$ to another label $L'$, we define
\[
c(E) = d_O(i(L), i(L')).
\]
That is, the cost of applying an operation $E$ of type (v) is computed as the distance from the type representing $L$ to the type representing $L'$ with respect to the metric $d_O$ on $T$.  The resulting metric $d$ on $W_s^\bullet$ then depends on the structure of the olog $O$ and the metric $d_O$ on the set $T$.
\end{eg}

As we will see in the next section, our definition of $d(W,W')$ utilizes properties of wiring diagrams and ologs that are not considered in usual definitions graph edit distance between two graphs.

\section{Example - comparing an analogy}\label{sec:compareanalogy}  We can now use elementary edit operations on skeleton wiring diagrams to quantify analogy between different concepts.

Suppose we want to compare the concept `an electric car charging station' and `a bus'.  In everyday language, we could say that these two concepts are analogous in the sense that both are physical entities capable of altering a characteristic of another physical entity.  That is, in order to determine the analogy between an electric car charging station and a bus, we must first spell out what we mean by these two concepts, such as:
\begin{itemize}
    \item[($S_1$)] An electric car charging station $s$ is a physical entity that increases the battery level of an electric car $c$, when the car is connected to the charging station.
    \item[($S_2$)] A bus $b$ is a physical object that can alter the location of a person $p$ when $p$ is inside  $b$.
\end{itemize}
Thus an electric car is an object that alters the characteristic `battery level' of an electric car, while a bus is an object that alters the characteristic `location' of a person.  To capture  this analogy mathematically, we need to represent these concepts as wiring diagrams.  We can think of wiring diagrams as giving a ``coordinate system'' for representing concepts such as a car charger or a bus, on which we can mathematically compare these concepts and quantify their similarity.

\paragraph[Sensing functions and wiring diagrams]\label{para:steps-sf}  For any electric car charging station $s$ and any electric car $c$, we will write $s \vDash c$ (resp.\ $s \nvDash c$) to mean `$c$ is connected to $s$' (resp.\ `$c$ is not connected to $s$').  We then define the sensing function $C_{s,c} : \{\bullet\} \to \{0,1\}$  by declaring 
\[
C_{s,c} (\bullet) = \begin{cases} 0 &\text{ if $s \nvDash c$} \\
1 &\text{ if $s \vDash c$} \end{cases}
\]
and subsequently a  `numerical derivative' 
\[
dC_{s,c}(\bullet) = (\text{current value of }C_{s,c})-(\text{value of }C_{s,c}\text{ five seconds ago}).
\]
Also, we define the sensing function $B_c : \{\bullet\} \to [0,100]$ that measures the battery level, as a percentage, of the electric car $c$.  If we model $B_c$ as a differentiable function over time $t$, we can take its derivative  $B_c'=\frac{dB_c}{dt}$ and thus define the sensing function $B_c^+ : \{\bullet\} \to \{0,1\}$ where
\[
B_c^+ (\bullet) = \begin{cases} 0 &\text{ if $B_c'\leq 0$} \\
1 &\text{ if $B_c'>0$} \end{cases}.
\]
Using the formulation in $S_1$, we can now represent the concept of an electric car charging station using the  wiring diagram  $W_1$ in  \eqref{fig:olog28}.

\begin{equation}\label{fig:olog28}
\begin{tikzcd}
	{W_1: } && {\begin{matrix}\bullet\\ (dC_{s,c},\bullet,1)\end{matrix}} && {\begin{matrix}\bullet\\ (B_c^+,\bullet,1)\end{matrix}}
	\arrow[shift left=2, from=1-3, to=1-5]
\end{tikzcd}\end{equation}

In plain language, this wiring diagram says the following: after an electric car $c$ is connected to a charging station $s$, the battery level of $c$ starts to increase.  Alternatively, we can use the wiring diagram in  \eqref{fig:olog34} to represent the concept of an electric car charging station.

\begin{equation}\label{fig:olog34}
\begin{tikzcd}
	{W_1': } && {\begin{matrix}\bullet\\ (C_{s,c},\bullet,0)\end{matrix}} && {\begin{matrix}\bullet\\ (C_{s,c},\bullet,1)\end{matrix}} && {\begin{matrix}\bullet\\ (B_c^+,\bullet,1)\end{matrix}}
	\arrow[shift left=2, from=1-3, to=1-5]
	\arrow[shift left=2, from=1-5, to=1-7]
\end{tikzcd}\end{equation}

Next, for any bus $b$ and any human $p$, we will write $b \succ p$ (resp.\ $b \nsucc p$) to mean  `$p$ is inside $b$' (resp.\ $p$ is not inside $b$').  This allows us to define the sensing function $T_{b,p} : \{\bullet\} \to \{0,1\}$ where
\[
 T_{b,p} (\bullet) = \begin{cases} 0 &\text{ if $b \nsucc p$} \\
1 &\text{ if $b \succ p$} \end{cases}.
\]
We also define a numerical derivative of $dT_{b,p}$ similarly to $dC_{s,c}$.  In addition, we  define a sensing function $L_p : \{\bullet\} \to [-90,90]\times [-180,180]$ that keeps track of the location of $p$ at any time $t$ as a pair $L_p(\bullet)=(x,y)$ of latitudinal and longitudinal coordinates $x$ and $y$.  Assuming $L_p$ is a smooth function with respect to $t$, we can define its second derivative $L_p'' = \frac{d^2 L_p}{dt^2}$ and subsequently the sensing function $A_p : \{\bullet\} \to \{0,1\}$ via
\[
 A_p (\bullet) = \begin{cases} 0 &\text{ if $|L_p''|= 0$} \\
1 &\text{ if $|L_p''|> 0$} \end{cases}.
\]
We can also define a differentiable function $D_p : \{\bullet \} \to \mathbb{R}$ that measures, at any point in time $t$, the distance travelled by $p$ since $t=t_0$, where $t_0$ is some fixed value.  Writing $D_p' =  \frac{dD_p}{dt}$, we can then form the sensing function $M_p : \{\bullet\} \to \{0,1\}$ such that
\[
M_p (\bullet) = \begin{cases} 0 &\text{ if $|D_p'|=0$} \\
1 &\text{ if $|D_p'|>0$} \end{cases}.
\]
Using the formulation $S_2$, we can now represent the concept of a bus using the wiring diagram $W_2$ in  \eqref{fig:olog29}.

\begin{equation}\label{fig:olog29}
\begin{tikzcd}
	{W_2: } && {\begin{matrix}\bullet\\ (dT_{b,p},\bullet,1)\end{matrix}} && {\begin{matrix}\bullet\\ (A_p,\bullet,1)\end{matrix}} && {\begin{matrix}\bullet\\ (M_p,\bullet,1)\end{matrix}}
	\arrow[shift left=2, from=1-3, to=1-5]
	\arrow[shift left=2, from=1-5, to=1-7]
\end{tikzcd}\end{equation}

In everyday language, this wiring diagram says that the concept of a bus  is characterised by the following sequence of events: a person enters a bus, the bus begins moving, resulting in the location of the person changing.   Alternatively, we can use the wiring diagram in  \eqref{fig:olog35} to represent the same concept.

\begin{equation}\label{fig:olog35}
\begin{tikzcd}
	{W_2': } && {\begin{matrix}\bullet\\ (T_{b,p},\bullet,0)\end{matrix}} && {\begin{matrix}\bullet\\ (T_{b,p},\bullet,1)\end{matrix}} && {\begin{matrix}\bullet\\ (A_p,\bullet,1)\end{matrix}} && {\begin{matrix}\bullet\\ (M_p,\bullet,1)\end{matrix}}
	\arrow[shift left=2, from=1-5, to=1-7]
	\arrow[shift left=2, from=1-7, to=1-9]
	\arrow[shift left=2, from=1-3, to=1-5]
\end{tikzcd}\end{equation}

\paragraph[Ologs]\label{para:eg-step2-ologs} Using the wiring diagram in  \eqref{fig:olog28} as a proxy for the concept of an electric car charging station, and that in  \eqref{fig:olog29} as a proxy for the concept of a bus, we can now attempt to calculate a distance between these two wiring diagrams using the method proposed in Section \ref{sec:distanceedit-WD}.  We will merely compute an upper bound of the distance by finding a third wiring diagram $W_3$ that is connected to both $W_1$ and $W_2$ via edit paths.  Diagram $W_3$ will represent an abstract process that accounts for commonalities between $W_1$ and $W_2$.  The labels in $W_3$ will make use of abstract concepts that give a connection between the concepts appearing in labels of $W_1$ and $W_2$; all these concepts will also be related via ologs. 

We begin by constructing  an olog as in  \eqref{fig:olog26}.

\begin{equation}\label{fig:olog26}
\begin{tikzcd}[scale cd=0.46]
	&&& {P_0: \boxed{\begin{matrix}\text{a pair } (y,f) \text{ where}\\ y \text{ is an entity,} \\ f : \{y\} \to \{0,1\} \\ \text{is a sensing function,} \\ \text{and }f(y)=0  \end{matrix}}} && {P: \boxed{\begin{matrix}\text{a pair } (y,f) \text{ where}\\ y  \text{ is an entity, and} \\ f : \{y\} \to \{0,1\} \\ \text{is a sensing function}  \end{matrix}}} \\
	&& {P_1: \boxed{\begin{matrix}\text{a pair } (y,f) \text{ where}\\ y \text{ is an entity,} \\ f : \{y\} \to \{0,1\} \\ \text{is a sensing function,}  \\ \text{and }f(y)=1  \end{matrix}}} \\
	& {A_0: \boxed{\begin{matrix}\text{a pair } (c,B_c^+) \text{ where}\\ c \text{ is an electric car, and} \\ B_c^+ = 0   \end{matrix}}} &&& {A: \boxed{\begin{matrix}\text{a pair } (c,B_c^+) \text{ where}\\ c \text{ is an electric car}   \end{matrix}}} \\
	{A_1: \boxed{\begin{matrix}\text{a pair } (c,B_c^+) \text{ where}\\ c \text{ is an electric car, and} \\ B_c^+ = 1   \end{matrix}}} &&& {\{0\}} && {\{0,1\}} \\
	&& {\{1\}}
	\arrow["e"{pos=0.3}, from=1-6, to=4-6]
	\arrow["{i_0}", from=4-4, to=4-6]
	\arrow[from=1-4, to=1-6]
	\arrow[from=1-4, to=4-4]
	\arrow["j", from=3-5, to=1-6]
	\arrow["ej", from=3-5, to=4-6]
	\arrow[from=3-2, to=3-5]
	\arrow[from=3-2, to=4-4]
	\arrow["{i_1}"', from=5-3, to=4-6]
	\arrow[from=4-1, to=3-5]
	\arrow[from=4-1, to=5-3]
	\arrow[from=2-3, to=5-3]
	\arrow[from=2-3, to=1-6]
\end{tikzcd}\end{equation}

To build this olog, we begin with the type $\lceil$a pair $(y,f)$ where $y$ is an entity, and $f$ is a $\{0,1\}$-valued sensing function that can be applied to $y$$\rceil$, which we denote by $P$.  We define the aspect $e$ to be the `evaluation map' that maps $(y,f)$ to the  number $f(y)$.  Then, we can construct the subtype $A$ that represents all the pairs of the form $(c,B_c^+)$ where $c$ is an electric car. (Recall that a type $T'$ is a \emph{subtype} of another type $T$ in an olog if every instance of $T'$ is also an instance of $T$.) That is, an instance of $A$ is a pair where the second coordinate is already fixed as the sensing function $B_c^+$ for the electric car $c$.  We can then define $A_0$ as the fiber product of $ej$ and $i_0$, $A_1$ as the fiber product of $ej$ and $i_1$, $P_0$ as the fiber product of $e$ and $i_0$, and $P_1$ as the fiber product of $e$ and $i_1$.  This way, we can take the type $A$ to be a representation of the concept $B_c^+$ in an olog, and take the type $A_1$ to be a representation of the concept $B_c^+=1$, which corresponds to the label $(B_c^+,\bullet,1)$ in the wiring diagram $W_1$.  Note that the instances of $P$ are in 1-1 correspondence with sensing functions $f^y: \{\bullet\} \to \{0,1\}$ that depend on the entity $y$, so we can use $P$ as a type that represents the concept of an arbitrary sensing function $f^y : \{\bullet\} \to \{0,1\}$ that tracks some characteristic of  some entity $y$.

Next, we can form an olog as in  \eqref{fig:olog25}.

\begin{equation}\label{fig:olog25}
\begin{tikzcd}[scale cd=0.4]
	&&& {T_0: \boxed{\begin{matrix} \text{a triple }(x,y,\thicksim)\text{ where} \\ x, y \text{ are entities,} \\ \thicksim \text{ is a relation between entities} \\ \text{and }x \nsim y \end{matrix}} } && {T: \boxed{\begin{matrix} \text{a triple }(x,y,\thicksim)\text{ where} \\ x, y \text{ are entities, and} \\ \thicksim \text{ is a relation between entities} \end{matrix}} } \\
	&& {T_1: \boxed{\begin{matrix} \text{a triple }(x,y,\thicksim)\text{ where} \\ x, y \text{ are entities,} \\ \thicksim \text{ is a relation between entities} \\ \text{and }x \thicksim y \end{matrix}} } \\
	& {G_0: \boxed{\begin{matrix} \text{a pair}(s,c)\text{ where} \\ s \text{ is an electric car} \\ \text{charging station,} \\ c \text{ is an electric car,} \\ \text{and }s\nvDash c  \end{matrix}} } &&& {G: \boxed{\begin{matrix} \text{a triple }(s,c,\vDash)\text{ where} \\ s \text{ is an electric car} \\ \text{charging station, and} \\ c \text{ is an electric car}  \end{matrix}} } \\
	{G_1: \boxed{\begin{matrix} \text{a pair}(s,c)\text{ where} \\ s \text{ is an electric car} \\ \text{charging station,} \\ c \text{ is an electric car,} \\ \text{and }s\vDash c  \end{matrix}} } \\
	&&& {\{0\}} && {\{0,1\}} \\
	&& {\{1\}} \\
	\\
	\\
	\\
	&&&&&&&&&& {}
	\arrow["{i_0}", from=5-4, to=5-6]
	\arrow["{i_1}"', from=6-3, to=5-6]
	\arrow["q", from=1-6, to=5-6]
	\arrow["{k_0}", from=1-4, to=1-6]
	\arrow[from=1-4, to=5-4]
	\arrow["{k_1}"{pos=0.6}, from=2-3, to=1-6]
	\arrow[from=2-3, to=6-3]
	\arrow["{j'}", from=3-5, to=1-6]
	\arrow["qj", from=3-5, to=5-6]
	\arrow[from=3-2, to=3-5]
	\arrow[from=3-2, to=5-4]
	\arrow[from=4-1, to=3-5]
	\arrow[from=4-1, to=6-3]
\end{tikzcd}\end{equation}

We begin by defining the type $\lceil$a triple $(x,y,\thicksim)$ where $x,y$ are entities, and $\thicksim$ is a relation between entities$\rceil$, denoted $T$, and the subtype $G$ that represents triples of the form $(x,y,\vDash)$, where $\vDash$ is the `is plugged into' relation from earlier.  The aspect $j'$ is the inclusion from $G$ into $T$, while $q$ is the aspect that takes a triple $(x,y,\thicksim)$ to the value $1$ (resp.\ $0$) if $x \thicksim y$ (resp.\ $x \nsim y$).  As before, $i_0$ and $i_1$ denote the respective set inclusions.  Then, we define $T_0$ as the fiber product of $q$ and $i_0$, $T_1$ as the fiber product of $q$ and $i_1$, $G_0$ as the fiber product of $qj'$ and $i_0$, and $G_1$ as the fiber product of $qj'$ and $i_1$.  Now  we can use the types $G_1, G_0$  as representations for the concepts defined by the labels $(C_{s,c},\bullet, 1)$ and  $(C_{s,c},\bullet, 0)$.  

Note that for an arbitrary relation $\thicksim$ between entities and any two entities $x$ and $y$, we can define a sensing function $F_{x,y,\thicksim} : \{\bullet\} \to \{0,1\}$ that gives the same value as $q$, i.e.\ $F_{x,y,\thicksim}(\bullet)$ equals $1$ (resp.\ $0$) when $x \thicksim y$ (resp.\ $x \nsim y$).

\paragraph[Elementary edit operations]\label{para:steps-eeo}   We give two different approaches to calculating the distance between the concept of an `electric car charging station' and a `bus', depending on the choices of wiring diagrams and cost functions along the way.

\subparagraph[Approach 1] Let us use wiring diagrams $W_1$ and $W_2$ as formulations of $S_1$ and $S_2$, respectively. The two  wiring diagrams $W_1$ and $W_2$ are  related via  elementary edit operations on wiring diagrams as in  Figure \ref{fig:olog31}.  Below, we use $\Rightarrow$ to denote an elementary edit operation so as to better  distinguish them from the arrows within wiring diagrams.

\begin{figure}[h]
    \centering
    \includegraphics[scale=0.5]{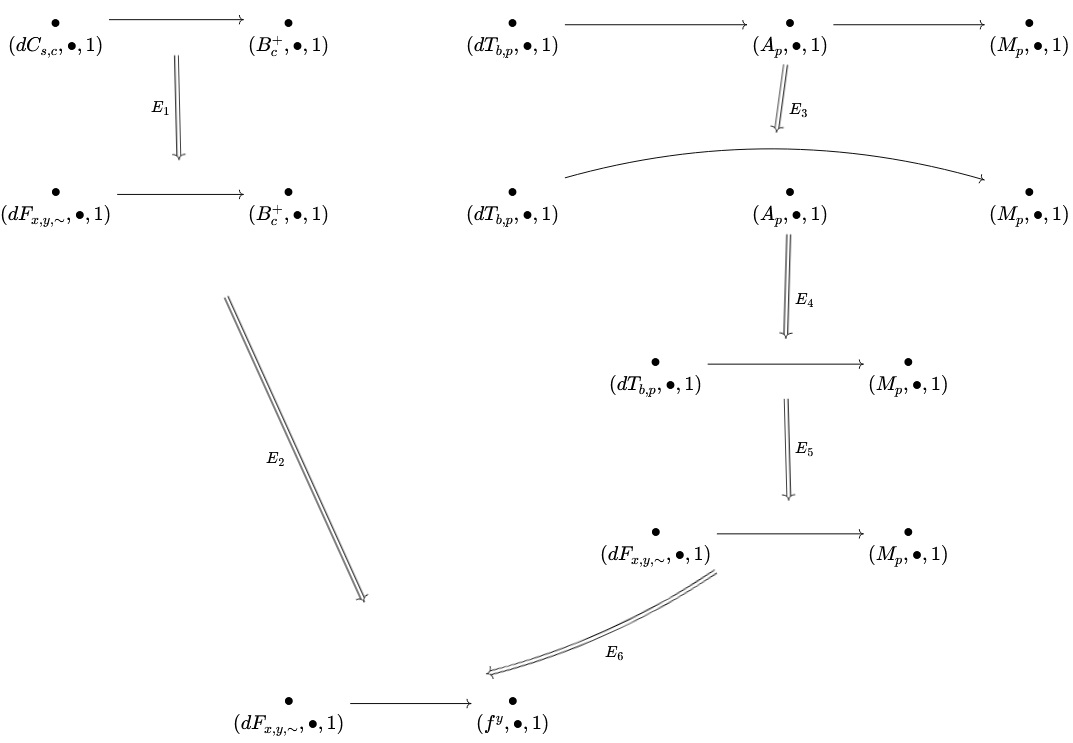}
    \caption{}
    \label{fig:olog31}
\end{figure}

In Figure \ref{fig:olog31}, $E_1, E_2, E_5, E_6$ are all operations of type (v) in the sense of Section \ref{sec:distanceedit-WD}, i.e.\ each of them is just a change of a single label;  in all these instances, we are changing a label to a more abstract label.  The operation $E_3$ is of type (viii) - it corresponds to an irreducible morphism in the category $\mathcal{R}(V)$ where $V$ is the set of vertices in $W_2$.  The operation $E_4$ is of type (ii), where the vertex with a single label $(A_p,\bullet,1)$ is deleted.  If write $E_i^{-1}$ for the inverse of an elementary edit operation $E^i$, then each $E_i$ is again an elementary edit operation and  $W_1$ is transformed into $W_2$ via the sequence of operations
\[
(E_1, E_2, E_6^{-1}, E_5^{-1}, E_4^{-1}, E_3^{-1}).
\]
Now for any cost function $c : \mathrm{EEO}(W_s^\bullet) \to \mathbb{R}_{>0}$,  we  obtain  an upper bound for the distance $d(W_1, W_2)$ using the metric $d$ from Lemma \ref{lem:def:distWD}:
\[
d(W_1, W_2) \leq \sum_{i=1}^2 c(E_i)+ \sum_{i=3}^6 c(E_i^{-1}) .
\]

\subparagraph[Approach 2]  Let us use $W_1'$ and $W_2'$ as representations of $S_1$ and $S_2$, respectively.  In this case, no wiring diagram labels are defined using numerical derivatives of sensing functions, and $W_1'$ and $W_2'$ are related via elementary edit operations as shown in Figure \ref{fig:olog33}.  We will also make use of the ologs in  \eqref{fig:olog26} and \eqref{fig:olog25} more directly in defining our cost function $c$ for the metric on $W_s^\bullet$.
\begin{figure}[h]
    \centering
    \includegraphics[scale=0.55]{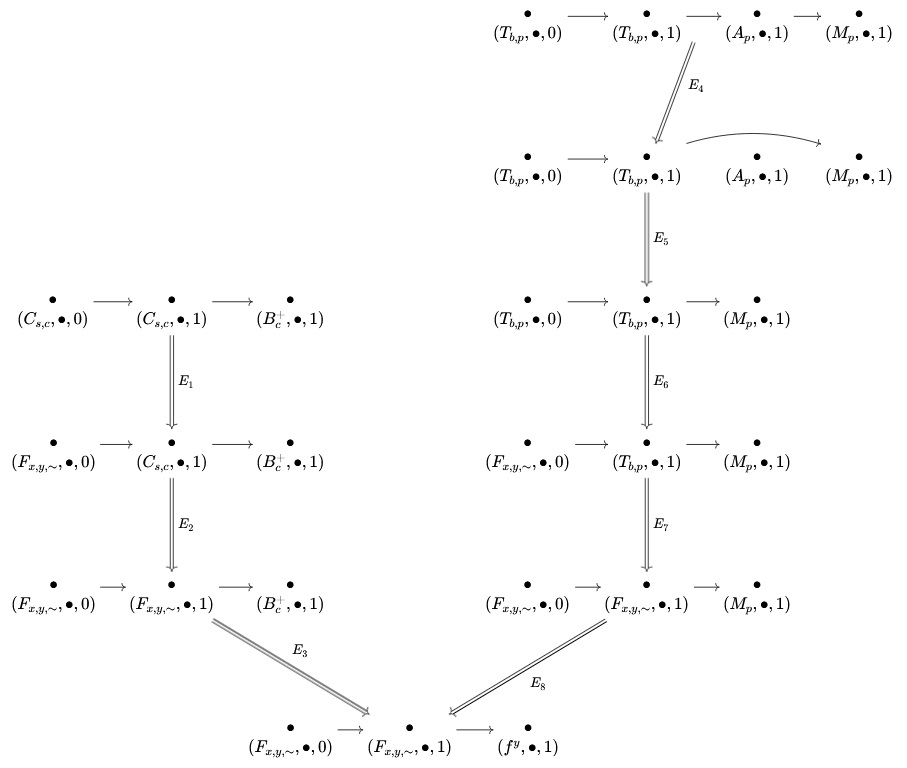}
    \caption{}
    \label{fig:olog33}
\end{figure}

In Figure \ref{fig:olog33}, the operation $E_4$ is of type (viii) while $E_5$ is of type (ii).  On the other hand, the operations $E_1, E_2, E_3, E_6, E_7, E_8$ are all of type (v);  each of these six operations involves changing the sensing function in a label to a different sensing function.  For example, $E_1$ involves changing $(C_{s,c},\bullet, 0)$ to $(F_{x,y,\thicksim},\bullet,0)$, while $E_3$ involves changing $(B_c^+, \bullet, 1)$ to $(f^y, \bullet, 1)$.  Note that all the labels involving $C_{s,c}, B_c^+, F_{x,y,\thicksim}, f^y$ in  Figure \ref{fig:olog33} are represented by types in the ologs  in  \eqref{fig:olog26} and \eqref{fig:olog25}:

\begin{center}
\begin{tabular}{|c | c |} 
 \hline
 Label & Type  \\  
 \hline\hline
 $(C_{s,c},\bullet,0)$ & $G_0$   \\ 
 \hline
 $(C_{s,c},\bullet,1)$ & $G_1$  \\
 \hline
 $(B_c^+,\bullet,1)$ & $A_1$  \\
 \hline
 $(F_{x,y,\thicksim},\bullet,0)$ &  $T_0$ \\
 \hline
 $(f^y,\bullet,1)$ & $P_1$  \\  
 \hline
\end{tabular}
\end{center}

Using constructions similar to those in  \eqref{fig:olog26} and \eqref{fig:olog25}, we could expand these ologs to contain types corresponding to all the other labels in Figure \ref{fig:olog33}, too.  Combining these ologs into a single olog $O$,  then choosing a cost function $c_O$ on the edges underlying $O$ and proceeding  as in Example \ref{eg:metricWsbwithologdist}, we obtain a metric $d$ on $W_s^\bullet$ that utilizes the ologs in  \eqref{fig:olog26} and \eqref{fig:olog25} in calculating $d(W_1', W_2')$.  From  Figure \ref{fig:olog33}, we now have the upper bound for $d(W_1', W_2')$
\[
d(W_1', W_2') \leq \sum_{i=1}^3 c(E_i) + \sum_{i=4}^8 c(E_i^{-1}).
\]

\begin{rem}
\begin{itemize}
    \item[(1)] Whether we take Approach 1 or Approach 2 above, the actual distance between the two wiring diagrams being compared would depend on the  olog being used to represent all the relevant concepts and the specific elementary edit operations allowed.  For example, in Approach 1, if we had allowed an arbitrary change of label in elementary edit operations of type (v), then there would be such an operation connecting the second diagram in the left column and the third diagram in the right column in Figure \ref{fig:olog31}.  The specific elementary edit operations of type (v)  used in Figures \ref{fig:olog31} and \ref{fig:olog33} make the point, that the two concepts we are trying to connect (`electric car charging station' and `bus') can both be connected to a more abstract concept (represented by the wiring diagram at the bottom in either  Figure \ref{fig:olog31} or Figure \ref{fig:olog33}).  In particular, the operation $E_3$ in Figure \ref{fig:olog31} and the operation $E_4$ in Figure \ref{fig:olog33}, both of which are morphisms in some category $\mathcal{R}(V)$, make mathematically precise what it means for a concept to be ``more abstract'' than another.
    \item[(2)] One could argue that, strictly speaking, some of the wiring diagrams in Figures \ref{fig:olog31} and \ref{fig:olog33} do not satisfy our definition of wiring diagrams (Definition \ref{def:WD}) because, in WD1, we require that the first argument of every vertex label be a specific sensing function, whereas entries such as $dF_{x,y,\thicksim}$ and $f^y$ in Figures \ref{fig:olog31} and \ref{fig:olog33} are `generic'  sensing functions.  We can get around this technical issue by extending the definition of wiring diagrams and allowing the arguments of vertex labels to be types in an olog.  This will be explored in a sequel to this article.
\end{itemize}
\end{rem}

\paragraph[Summary]\label{para:summaryofcalc} We now give a summary of the steps that one can follow in order to  compute the distance between pairs of concepts in a given application domain.  Suppose the concepts we are concerned with are elements of an indexed set $\{N_i\}_{i\in I}$.  Then one can perform the following tasks in the listed order:
\begin{enumerate}
    \item[(1)] Define the relevant sensing functions.
    \item[(2)] Define  wiring diagrams $W_i$ that represent the concepts $N_i$.
    \item[(3)] Construct an olog (or ologs) containing types that correspond to all the labels in  the wiring diagrams $W_i$  (e.g.\ see  \ref{sec:WDandologs} and \ref{sec:relnsfs}).
    \item[(4)] Decide on a list of acceptable elementary edit operations on wiring diagrams.  For example, one may wish to restrict the kinds of allowed operations of type (v) in the list in \ref{sec:distanceedit-WD}.
    \item[(5)] Decide on a cost function $c$ in the definition of the metric on wiring diagrams in Lemma \ref{lem:def:distWD}.  More specifically, one needs to decide on the cost of each elementary edit operation, such as the cost of an operation of type (v) - see Example \ref{eg:metricWsbwithologdist}. 
    \item[(6)] For any two distinct concepts $N_i, N_j$, calculate their distance $d(N_i, N_j)$ using the definition in Lemma \ref{lem:def:distWD}.  Each possible edit path from $N_i$ to $N_j$ would constitute a `justification', or a mathematical breakdown of the analogy between concept $N_i$ and concept $N_j$.
\end{enumerate}

For the main example in this section, Steps (1) and (2) were implemented in  \ref{para:steps-sf}, Steps (3) was implemented in \ref{para:eg-step2-ologs}, while Steps (4) through (6) were implemented in \ref{para:steps-eeo}.

\section{Future directions}\label{sec:futureDs}

In this article, we first recalled how ologs can be used to represent abstract concepts.  Then we define the concept of wiring diagrams where labels at vertices correspond to types in an olog.  Wiring diagrams allow us to represent concepts corresponding to temporal processes, which may not be so easily represented using ologs alone.  We can think of wiring diagrams as giving a coordinate system, or a state space on which one can develop a theory of problem-solving.  This direction will be explored in a sequel to this article.

As mentioned in \ref{para:intro-WD}, the term `wiring diagram' has also been defined and studied as operads in works such as \cite{SR-WDdisc,spivak2013operad,VSL,yau2018operads}.  The wiring diagrams as defined in this article certain show features of self-similarity - under appropriate assumptions,  one can replace any vertex in a wiring diagram (along with its state vector) by a wiring diagram to obtain a more complicated wiring diagram.  It would be worthwhile to reconcile the definition of wiring diagrams in this article with those in the aforementioned works.  In the present article, we refrained from doing so in order to keep our theory accessible to a wider audience.

Example \ref{eg:WDgraphsonsameV} hinted at the complexity that can be encoded within the underlying graphs of wiring diagrams.  For example, a wiring diagram of the form \eqref{fig:olog18} may be an indication of the social behavior of collaboration.  This opens up a host of questions to be answered.  For example, given a sequence of events over time, what are the possible wiring diagrams that possess these events as the state vectors, and how many are there?  Mathematically, this is related to the problem of enumerating all the preorders or partial orders on a set of given objects, and perhaps related to the notion of graph fibrations \cite{boldi2002fibrations}.  One can also ask if wiring diagrams can be used to classify behaviors, whether in the context of biology (behaviors of different species), social science (behaviors of humans or organizations), finance (behaviors of markets).

Lastly, the definition of wiring diagrams we adopted  in this paper applies  to any type of data that admits a fibration into a linearly ordered set - the concept of ordering among the state vectors in a wiring diagram comes from condition WD2 in Definition \ref{def:WD}.  In  all the wiring diagrams we considered in this paper, the state vectors always corresponded to events that can be partially ordered with respect to time (i.e.\ whether one event is required to occur before another).  Nonetheless, one can just as well consider wiring diagrams where the ordering is given by causation, for example, as in the case of mathematical proofs.  As shown in examples in \cite[Sections 6.6-6.7]{SpivakKent} (see also \cite{Petros1}), some mathematical definitions can  be expressed via ologs, after which mathematical lemmas can be expressed as commutativity of diagrams within the olog.  One could potentially think of a mathematical proof as a wiring diagram where the state vectors correspond to various `milestones' in the proof, and where arrows are defined using causation among the milestones.  One could then make precise what we mean when we say two mathematical proofs are `similar', or that the argument of one proof in a specific context `carries over' in a different context.

\bibliography{refsMR}{}

\begin{thebibliography}{10}

\bibitem{Awodey}
Steve Awodey.
\newblock {\em Category Theory}.
\newblock Oxford University Press, Inc., USA, 2nd edition, 2010.

\bibitem{boldi2002fibrations}
Paolo Boldi and Sebastiano Vigna.
\newblock Fibrations of graphs.
\newblock {\em Discrete Mathematics}, 243(1-3):21--66, 2002.

\bibitem{brommer2015categorical}
Dieter~B Brommer, Tristan Giesa, David~I Spivak, and Markus~J Buehler.
\newblock Categorical prototyping: Incorporating molecular mechanisms into 3d printing.
\newblock {\em Nanotechnology}, 27(2):024002, 2015.

\bibitem{drozd-etal-2016-word}
Aleksandr Drozd, Anna Gladkova, and Satoshi Matsuoka.
\newblock Word embeddings, analogies, and machine learning: Beyond king - man + woman = queen.
\newblock In Yuji Matsumoto and Rashmi Prasad, editors, {\em Proceedings of {COLING} 2016, the 26th International Conference on Computational Linguistics: Technical Papers}, pages 3519--3530, Osaka, Japan, December 2016. The COLING 2016 Organizing Committee.

\bibitem{GXTL}
X.~Gao, B.~Xiao, D.~Tao, and X.~Li.
\newblock A survey of graph edit distance.
\newblock {\em Pattern Analysis and Applications}, 13:113--129, 2010.

\bibitem{giesa2012category}
Tristan Giesa, David~I Spivak, and Markus~J Buehler.
\newblock Category theory based solution for the building block replacement problem in materials design.
\newblock {\em Advanced Engineering Materials}, 14(9):810--817, 2012.

\bibitem{handGT}
Jonathan~L. Gross, Jay Yellen, and Ping Zhang.
\newblock {\em Handbook of Graph Theory, Second Edition}.
\newblock Chapman \& Hall/CRC, 2nd edition, 2013.

\bibitem{KRAWCZYK2018227}
Daniel~C. Krawczyk.
\newblock Chapter 10 - analogical reasoning.
\newblock In Daniel~C. Krawczyk, editor, {\em Reasoning}, pages 227--253. Academic Press, 2018.

\bibitem{maclane:71}
Saunders MacLane.
\newblock {\em Categories for the Working Mathematician}.
\newblock Springer-Verlag, New York, 1971.
\newblock Graduate Texts in Mathematics, Vol. 5.

\bibitem{Petros1}
P.~Mavromichalis.
\newblock Proofs in group theory in the context of ologs, 2023.
\newblock Master's thesis at California State University, Northridge.

\bibitem{gapGEDKM}
Michel Neuhaus and Horst Bunke.
\newblock {\em Bridging the Gap Between Graph Edit Distance and Kernel Machines}.
\newblock World Scientific Publishing Co., Inc., USA, 2007.

\bibitem{PerezSpivak}
M.~A. P\'{e}rez and D.~I. Spivak.
\newblock Toward formalizing ologs: linguistic structures, instantiations, and mappings.
\newblock Preprint. arXiv:1503.08326 [math.CT], 2015.

\bibitem{petersen-van-der-plas-2023-language}
Molly Petersen and Lonneke van~der Plas.
\newblock Can language models learn analogical reasoning? investigating training objectives and comparisons to human performance.
\newblock In Houda Bouamor, Juan Pino, and Kalika Bali, editors, {\em Proceedings of the 2023 Conference on Empirical Methods in Natural Language Processing}, pages 16414--16425, Singapore, December 2023. Association for Computational Linguistics.

\bibitem{Riesen:1413749}
K.~Riesen and H.~Bunke.
\newblock {\em Graph Classification and Clustering Based on Vector Space Embedding}.
\newblock Series in Machine Perception and Artificial Intelligence. World Scientific, Singapore, 2010.

\bibitem{SR-WDdisc}
D.~Rupel and D.~I. Spivak.
\newblock The operad of temporal wiring diagrams: formalizing a graphical language for discrete-time processes.
\newblock Preprint. arXiv:1307.6894 [math.CT], 2013.

\bibitem{Serra1}
F.~Serratosa.
\newblock Redefining the graph edit distance.
\newblock {\em SN Computer Science}, 2(438), 2021.

\bibitem{spivak2013operad}
D.~I. Spivak.
\newblock The operad of wiring diagrams: formalizing a graphical language for databases, recursion, and plug-and-play circuits.
\newblock 2013.

\bibitem{SpivakCTS}
David~I. Spivak.
\newblock {\em Category Theory for the Sciences}.
\newblock The MIT Press, 2014.

\bibitem{spivak2011category}
David~I Spivak, Tristan Giesa, Elizabeth Wood, and Markus~J Buehler.
\newblock Category theoretic analysis of hierarchical protein materials and social networks.
\newblock {\em PloS one}, 6(9):e23911, 2011.

\bibitem{SpivakKent}
David~I. Spivak and Robert~E. Kent.
\newblock Ologs: A categorical framework for knowledge representation.
\newblock {\em PLOS ONE}, 7(1), 01 2012.

\bibitem{VSL}
D.~Vagner, D.~I. Spivak, and E.~Lerman.
\newblock Algebras of open dynamical systems on the operad of wiring diagrams.
\newblock {\em Theory and Applications of Categories}, 30(51):1793--1822, 2015.

\bibitem{walters_1992}
R.~F.~C. Walters.
\newblock {\em Categories and Computer Science}.
\newblock Cambridge Computer Science Texts. Cambridge University Press, 1992.

\bibitem{wills2020metrics}
Peter Wills and Fran{\c{c}}ois~G Meyer.
\newblock Metrics for graph comparison: a practitioner’s guide.
\newblock {\em Plos one}, 15(2):e0228728, 2020.

\bibitem{WuYY}
Y.~Wu.
\newblock Gene ologs: a categorical framework for gene ontology.
\newblock Preprint. arXiv:1909.11210 [q-bio.GN], 2019.

\bibitem{yau2018operads}
Donald Yau.
\newblock {\em Operads of wiring diagrams}, volume 2192.
\newblock Springer, 2018.

\end{thebibliography}
\bibliographystyle{plain}

\end{document}